\documentclass[11pt]{article}
\usepackage[a4paper,margin=1in]{geometry}
\usepackage{mathtools,amssymb,amsthm}
\usepackage{graphicx}
\usepackage[numbers,sort&compress]{natbib}
\usepackage{hyperref}
\usepackage{mathtools}
\usepackage{algorithm}
\usepackage{algorithmic}
\usepackage{authblk}
\usepackage{tikz, multirow,xcolor}
\usetikzlibrary{arrows.meta,calc,positioning,decorations.pathreplacing,shapes.geometric}
\usepackage[numbers,sort&compress]{natbib}
\definecolor{hsnwBlue}{RGB}{42,97,160}
\definecolor{hsnwRed}{RGB}{165,64,58}
\definecolor{hsnwGray}{RGB}{105,112,120}
\providecommand{\Description}[2][]{}
\newcommand{\T}{\mathbb T}
\newcommand{\PP}{\mathbb P}
\newcommand{\E}{\mathbb E}
\newcommand{\R}{\mathbb R}
\newcommand{\cC}{\mathcal C}

\newcommand{\cG}{\mathcal G}
\newcommand{\cH}{\mathcal H}
\newcommand{\cF}{\mathcal F}
\newcommand{\cK}{\mathcal K}
\newcommand{\cN}{\mathcal N}
\newcommand{\1}{\mathbf 1}
\newcommand{\eps}{\varepsilon}
\newtheorem{maintheorem}{Theorem}
\newtheorem{theorem}{Theorem}[section]
\newtheorem{lemma}[theorem]{Lemma}
\newtheorem{proposition}[theorem]{Proposition}

\theoremstyle{definition}
\newtheorem{definition}[theorem]{Definition}

\title{The Geometry of Hierarchical Navigation: Accuracy and Query Cost for Point Process Input}

\author[1]{Shankar Bhamidi}
\author[2]{Souvik Dhara}
\author[3]{Lars Schroeder}
\author[3]{Clara Stegehuis}
\affil[1]{Department of Statistics and Operations Research, University of North Carolina at Chapel Hill, Chapel Hill, NC, USA}
\affil[2]{H. Milton Stewart School of Industrial and Systems Engineering, Georgia Institute of Technology, Atlanta, GA, USA}
\affil[3]{Mathematics of Operations Research, Faculty of Electrical Engineering, Mathematics and Computer Science, University of Twente, Enschede, The Netherlands}
\date{}

\begin{document}
\maketitle
\begin{abstract}
Large-scale information retrieval systems, including retrieval-augmented generation (RAG) and recommendation engines, widely use multi-layered hierarchical data structures for ultra-fast approximate nearest-neighbor search in high-dimensional vector spaces. However, the geometric conditions that ensure accurate and efficient greedy navigation remain poorly understood. In this work, we study the efficiency of greedy navigation on a hierarchy of proximity graphs constructed from \(n\) data points on the \(d\)-dimensional torus~$\mathbb{T}^d$. We identify a deterministic coverage condition under which, given any query $q\in \mathbb{T}^d$, greedy search returns a point within $(1+\varepsilon)$-factor of the distance to the closest point.
This coverage property holds with high probability when the data is distributed as a homogeneous Poisson process, a Hermitian determinantal process, or a bounded-density Cox process, as long as $d = o(\log n/\log \log n)$. Under the same assumptions, the expected number of greedy hops for a fixed query is \(O\!\left(\exp\!\left(\tfrac12 d\log d+O(d)\right)\log n\right)\), yielding logarithmic expected hop count in fixed dimension.
\end{abstract}

\section{Introduction}
\label{sec:introduction}
Ultra-fast approximate nearest-neighbor (ANN) search is a core computational component of retrieval-augmented generation (RAG)-based large language model (LLM) systems~\cite{Lewis2020}, large-scale recommendation engines~\cite{Covington2016}, image matching~\cite{Lowe2004}, and place recognition in robotics~\cite{An2019}. In a typical RAG pipeline, documents are split into passages and embedded as vectors. The query is similarly embedded, and ANN search retrieves passages with nearby embedding vectors. The retrieved passages are supplied alongside the query as context to an LLM, which uses this information to generate an answer~\cite{Lewis2020}. Similarly, semantic search retrieves text by similarity between learned representations~\cite{ReimersGurevych2019}, recommendation systems retrieve candidate items using user and item embeddings~\cite{ChenEtAl2022}, and image retrieval finds images with nearby feature vectors~\cite{JohnsonEtAl2017}. For large collections and high query rates, comparing each query with every stored vector is computationally expensive~\cite{Covington2016,JohnsonEtAl2017}, motivating nearest-neighbor search methods that trade exactness for faster retrieval.

A useful approach is to organize the data into a hierarchy of proximity graphs, as in Hierarchical Navigable Small World (HNSW) indexes~\cite{MalkovYashunin2020}.
Classical ANN algorithms obtain efficient $(1+\varepsilon)$-approximation guarantees through carefully designed data structures~\cite{IndykMotwani1998,arya1993approximate}.
HNSW instead organizes the data into nested layers generated by random sampling~\cite{MalkovYashunin2020}, and navigates proximity graphs constructed using nearest-neighbor heuristics.
Search starts in a sparse upper layer and proceeds toward the query through successively denser layers, combining coarse navigation with local refinement.
Compared with earlier greedy navigation methods, this combination separates navigation across distance scales and substantially improves empirical search performance, particularly on clustered data~\cite{MalkovYashunin2020}.
See Figure~\ref{fig:hierarchy} for an illustration.
Libraries such as Faiss~\cite{FaissIndexes} and hnswlib~\cite{Hnswlib} support this approach through implementations of hierarchical graph-based ANN search.
HNSW-based retrieval is used in documented RAG implementations built around Claude~3 Sonnet~\cite{AWSClaudeHNSW2024}, Llama~3~70B~\cite{AWSLlamaHNSW2024}, and GPT-4o~\cite{MicrosoftGPT4oHNSW2024}.
\emph{The geometric properties that make navigation on such hierarchical proximity graphs  accurate and efficient
are the focus of this work.}

Following the introduction of HNSW, substantial progress has been made in identifying graph properties that support accurate search~\cite{IndykXu2023,GollapudiEtAl2025,PengEtAl2023,LiYanLu2026}.
The search strategy and the hierarchy introduce additional challenges.
Even when a graph permits greedy routing between stored points, standard fixed-width beam search, which retains several promising candidates rather than following a single path, can fail to provide any finite approximation guarantee for arbitrary queries~\cite{AlJazzaziEtAl2025}. Experiments in~\cite{IndykXu2023} also show that, on carefully constructed datasets, practical HNSW implementations may need to examine a constant fraction of the data points to achieve high nearest-neighbor retrieval accuracy, even in two-dimensional Euclidean space.
These findings motivate understanding when simple greedy navigation succeeds on hierarchical proximity graphs under suitable assumptions on the data.

For hierarchical search based on $\varepsilon$-nets, HENN~\cite{DehghankarAsudeh2025} obtains search-cost guarantees when each upper layer intersects every sufficiently populated spherical shell of the layer below.
A complementary line of work analyzes navigation under probabilistic input models. Analyses of graph-based nearest-neighbor search consider independent uniform points on a sphere, including query models with a planted close neighbor~\cite{Laarhoven2018,ProkhorenkovaShekhovtsov2020}.
Navigation on point processes also has a substantial stochastic-geometric literature, establishing laws of large numbers, path length and hop count asymptotics, and deviation bounds for radial, directed, and cone-based routing~\cite{BaccelliBordenave2007,Bordenave2008,BonichonMarckert2011,GhoshJahnelJhawar2026}.
These works provide probabilistic context and motivation for our analysis of greedy navigation on an HNSW-inspired hierarchy.
\paragraph*{Our contribution.} We identify a deterministic directional coverage condition that ensures accurate greedy navigation on an HNSW-inspired hierarchy of proximity graphs. Under this condition, greedy search returns a
$(1+\varepsilon)$-approximate nearest neighbor simultaneously for every query, every nonempty layer, and every starting vertex.
We then show that this condition holds with high probability for homogeneous Poisson processes, Hermitian determinantal processes, and bounded-density Cox processes, encompassing independent, repulsive, and clustered point configurations.
In the dimension regime $d=o(\log n/\log\log n)$, we also show that, for a fixed query, the expected number of greedy hops is
\(
O\!\left(\exp\!\left(\tfrac12 d\log d+O(d)\right)\log n\right).
\)
Thus, a common deterministic coverage condition supports uniform accuracy guarantees and, together with probabilistic bounds on local point counts, yields logarithmic expected hop count in fixed dimension across these input models.
Experimental evaluation on these point process models  complements the theoretical analysis, illustrating how approximation accuracy varies with dimension and sector count. We also examine directional coverage and compare the routes produced by greedy navigation with shortest paths in the constructed hierarchy.

\section{Model and algorithm}
\label{sec:framework}

\subsection{Points with hierarchy}
We first describe the setup of the layered hierarchy.
Let $\T^d=\mathbb R^d/\mathbb Z^d$, with normalized Lebesgue measure, where $d=d_n\ge1$ may depend on $n$. For $x,y\in\T^d$, let $\delta(x,y)\in[-1/2,1/2)^d$ be the canonical lift of $x-y$ and write $d(x,y)=\|\delta(x,y)\|_2$. The torus diameter is $\Delta_d=\sqrt d/2$ and $c_d$ is the volume of the Euclidean unit ball. Because of that, on the torus, balls of radius~$r<1/2$ have volume $c_dr^d$. For $v \in \T^d$, we denote balls of radius $r$ by $B(v,r) = \{ x\in \T^d: d(x,v) \le r\}$. For nonzero vectors $x,y\in \R^d$, define the angle between them as $\angle(x,y):= \arccos(\langle x,y \rangle/(\|x\|_2\|y\|_2))$.

We begin with a finite set of distinct points $X_0\subset\T^d$ and a size parameter $n>1$. For a deterministic dataset, we take $n=|X_0|$ and for the point-process models below, $n$ specifies the intensity scale.
Fix $\rho\in(0,1)$ and assign each point $x\in X_0$ an independent geometric height $H_x\in\{0,1,2,\ldots\}$ with $\PP(H_x\ge\ell)=\rho^\ell$.
This is the independent layer-sampling mechanism used in HNSW~\cite{MalkovYashunin2020}. Here we use the deterministic cutoff $L$ on heights. This choice keeps the top-layer intensity scale of order $D\log n$. Each point is present in every layer up to its height, and we construct the hierarchy only up to layer $L$.
We call $X_0$ together with the height marks $(H_x)_{x\in X_0}$ the marked input.
To compare different layers, we introduce a reference scale for their local density. If layer \(\ell\) contains about \(n_\ell\) points, then the characteristic distance between nearby points is of order \(n_\ell^{-1/d}\). We define
\begin{equation}\label{eq:scale} 
n_\ell=n\rho^\ell, \qquad h_\ell= \Big( \frac{\log n}{n_\ell} \Big)^{1/d}.
\end{equation} The extra factor \(\log n\) is chosen so that balls of radius \(h_\ell\) contain logarithmically many points for fixed $d$, providing enough points for the large-deviation estimates used later. The condition on the deterministic hierarchy later will require that the distances to the closest points on layer $\ell$ are at most at scale $h_\ell$.
For $n>1$, define
\begin{equation}
L=\max\left\{0,\left\lfloor\log_{1/\rho}\frac{n}{D\log n}\right\rfloor\right\}, \qquad X_\ell=\{x\in X_0:H_x\ge\ell\},\quad 0\le\ell\le L.
 \label{eq:hierarchy}
\end{equation}
Here, $D>0$ controls the intensity scale of the top layer and may depend on $d$ and the fixed model parameters. Its choice is specified in Appendix~\ref{app:dimension}. For sufficiently large $n$, this yields $$D\log n\le n_L<(D/\rho)\log n,$$ so that the top layer contains logarithmically many points.

\begin{figure}[t]
\centering
\begin{tikzpicture}[>=Latex,x=.92cm,y=.92cm,font=\footnotesize]

\foreach \ell/\lab in {0/0,1/1,2/L}{
  \begin{scope}[shift={(0,1.2*\ell)}]
    \filldraw[fill=hsnwBlue!15, draw=none]
      (0,0)--(5.2,0)--(6.25,.85)--(1.05,.85)--cycle;
    \node[left] at (0,.18) {$X_{\lab}$};
    \foreach \id/\u/\v/\h in {
      A/.12/.25/2,B/.40/.50/2,C/.56/.72/1,W/.66/.65/0,
      S/.91/.45/0,D/.87/.18/0,E/.82/.90/0,
      F/.20/.75/1,G/.35/.18/0,H/.57/.15/1,I/.63/.82/0,
      J/.08/.60/0,K/.38/.82/2,M/.72/.22/0,N/.96/.83/0}{
      \ifnum\h<\ell\relax\else
        \coordinate (\id\ell) at ({5.2*\u+1.05*\v},{.85*\v});
        \fill[hsnwBlue!60] (\id\ell) circle (1.65pt);
      \fi
    }
    \coordinate (q\ell) at ({5.2*.91+1.05*.65},{.85*.65});
  \end{scope}
}

\draw[gray!40,densely dotted] (A0)--(A2);
\draw[gray!40,densely dotted] (K0)--(K2);
\draw[gray!40,densely dotted] (B0)--(B1);
\draw[hsnwRed,thick,->] (A2)--(B2);
\draw[hsnwRed,thick,dashed,->] (B2)--(B1);
\draw[hsnwRed,thick,->] (B1)--(C1);
\draw[hsnwRed,thick,dashed,->] (C1)--(C0);
\draw[hsnwRed,thick,->] (C0)--(W0);
\draw[hsnwRed,thick,fill=hsnwBlue!80] (W0) circle (2.3pt);
\fill[green!45!black] (S0) circle (1.8pt);
\foreach \ell in {0,1,2}{
 \node[star,star points=5,fill=hsnwBlue,inner sep=1.3pt,label={[hsnwBlue]above right:$q$}] at (q\ell) {};
}
\node[blue,right] at (W0) {\textcolor{hsnwRed}{$W_0$}};
\node[blue,below right] at (S0) {\textcolor{green!45!black}{$x_*$}};
\node[above left] at (A2) {entry};
\node[above] at (3.1,3.4) {Navigation on nested point sets};
\node[align=center] at (3.1,-.45) {\scriptsize Solid red: greedy hops within a layer.\\
\scriptsize Dashed red: descending at the same point.};

\node[above] at (10.25,3.4) {Layer 0 close-up};
\coordinate (Q) at (10.65,1.5);
\coordinate (W) at (9.15,1.5);
\coordinate (S) at (10.65,.3);
\coordinate (U) at (9.51,.3);
\coordinate (V) at (8.52,2.45);
\draw[hsnwRed!60,dashed] (Q) circle (1.5);
\draw[hsnwBlue,densely dotted] (Q) circle (1.2);
\draw[gray!65,->] (W)--(U);
\draw[gray!65,->] (W)--(V);
\fill[gray!75] (U) circle (1.6pt);
\fill[gray!75] (V) circle (1.6pt);
\draw[hsnwRed,thick] (W)--node[above] {$d(W_0,q)$}(Q);
\draw[hsnwBlue,thick] (Q)--node[right] {$d(x_*,q)$}(S);
\draw[hsnwRed,thick,fill=white] (W) circle (2.5pt);
\node[left] at (W) {\textcolor{hsnwRed}{$W_0$}};
\node[star,star points=5,fill=hsnwBlue,inner sep=2pt,label=above right:\textcolor{hsnwBlue}{$q$}] at (Q) {};
\fill[green!45!black] (S) circle (2.2pt);
\node[below] at (S) {\textcolor{green!45!black}{$x_*$}};
\node[align=center] at (10.25,-.5) {\scriptsize $d(W_0,q)\leq (1+\varepsilon) d(x_*,q).$\\
\scriptsize Gray outgoing arcs do not improve query distance.};
\end{tikzpicture}
\caption{Hierarchical navigation on two-dimensional point sets.
A star marks the same query $q$ on each layer. The search descends through copies of the same data point and ends at the marked data point $W_0$ on layer $0$.
The close-up shows that this approximate answer $W_0$ may not be the true nearest
neighbor $x_*$. Our theorem shows that $d(W_0,q)\leq (1+\varepsilon) d(x_*,q)$.}
\Description{Three stacked parallelograms contain nested two-dimensional point clouds. Red arrows show greedy moves and dashed vertical arrows show descent. A planar close-up shows the returned point W zero slightly farther from the query than the true nearest point x star, with a distance ratio of 1.25. The outgoing neighbors shown at W zero are farther from the query, so greedy search stops there.}
\label{fig:hierarchy}
\end{figure}
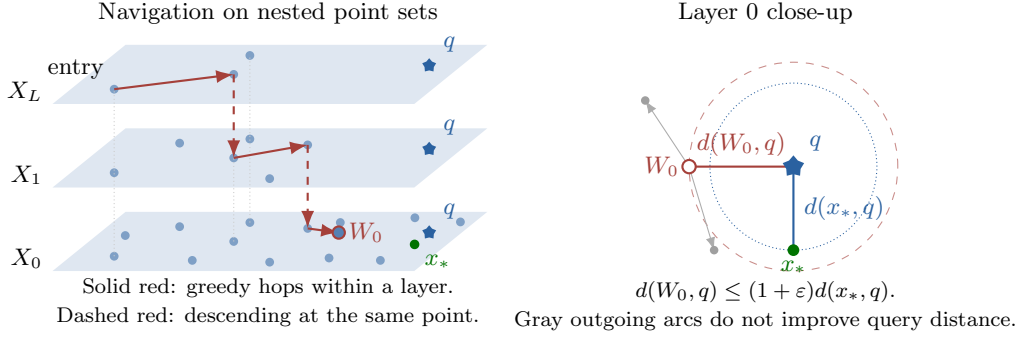

\subsection{Directional diversification and candidate pools}
\label{sec:directional}
We now show how we construct the graph to ensure that it supports effective greedy search in all potential directions (Algorithm~\ref{alg:build}). To achieve this, we partition the space around a vertex into sectors and aim to connect to one
representative per occupied sector. Directional neighbor selection also underlies classical cone-based graph constructions and search methods on a single graph, without the sampled hierarchy used here~\cite{Yao1982,Clarkson1987,Keil1988,arya1993approximate}.

Let $\mathbb S^{d-1}=\{u\in\mathbb R^d:\|u\|_2=1\}$ denote the unit sphere. Partition $\mathbb S^{d-1}$ into Borel measurable cells $\mathcal U_1,\ldots,\mathcal U_K$ whose boundaries have zero $(d-1)$-dimensional surface measure. Let $\sigma$ denote this surface measure. Write
\begin{equation}
 p_j=\frac{\sigma(\mathcal U_j)}{\sigma(\mathbb S^{d-1})},\qquad
 p_* =\min_jp_j>0,\qquad
 \omega=\max_j\operatorname{diam}_{\angle}(\mathcal U_j)<\pi/3.
 \label{eq:partition}
\end{equation}
Here, $\operatorname{diam}_{\angle}(\mathcal U_j)=\sup_{u,w\in\mathcal U_j}\arccos\langle u,w\rangle$ is the angular diameter of the cell.
For $d=1$, $\sigma$ is the counting measure on $\mathbb S^0=\{-1,1\}$ and the two
singleton cells have diameter zero. We then define the $j$-th sector at $v$ as
\[
 \cC_j(v)=\left\{x\ne v:
 \delta(x,v)/\|\delta(x,v)\|\in\mathcal U_j\right\}.
\] For random point processes, the heights are also independent of the point process. Let $C_{\rm pool}>0$ be a chosen constant, which may depend on $d$ and the fixed model parameters, and let $m=\lceil C_{\rm pool}\log n\rceil$. We now create the graph at each layer. For every vertex $v$ and layer $\ell$, let $P_\ell(v)$ contain the $\min\{m,|X_\ell|-1\}$ nearest points to $v$ in $X_\ell\setminus\{v\}$. For each occupied sector, add an arc from $v$ to the nearest point in $P_\ell(v)\cap\cC_j(v)$. If the nearest neighbor is not unique, choose one arbitrarily. Denote the resulting out-neighbors by $N_\ell^+(v)$. By construction, $|N_\ell^+(v)|\le K$. Whenever an arc points into a sector, its endpoint is a closest point to $v$ among all layer-$\ell$ points in that sector.

\begin{algorithm}[t]
\small
\caption{Construct the out-neighbor hierarchy}\label{alg:build}
\begin{algorithmic}[1]
\REQUIRE Finite point set $X_0\subset\T^d$; size parameter $n>1$; $\rho\in(0,1)$; $D,C_{\rm pool}>0$; sector partition $(\mathcal U_j)_{j=1}^K$;
\STATE Set $L$ as in \eqref{eq:hierarchy} and $m=\lceil C_{\rm pool}\log n\rceil$
\STATE Independently sample heights with $\PP(H_x\ge\ell)=\rho^\ell$ and form $(X_\ell)_{\ell=0}^L$
\FOR{$\ell=0,\ldots,L$}
 \FOR{$v\in X_\ell$}
  \STATE $N_\ell^+(v)\gets\varnothing$
  \STATE $P_\ell(v)\gets$ the $\min\{m,|X_\ell|-1\}$ nearest points to $v$ in $X_\ell\setminus\{v\}$
  \FOR{$j=1,\ldots,K$}
   \IF{$P_\ell(v)\cap\cC_j(v)\neq\varnothing$}
    \STATE Add $\arg\min_{u\in P_\ell(v)\cap\cC_j(v)}d(v,u)$ to $N_\ell^+(v)$
   \ENDIF
  \ENDFOR
 \ENDFOR
\ENDFOR
\RETURN $(X_\ell)_{\ell=0}^L$ and $\{N_\ell^+(v):0\le\ell\le L,\ v\in X_\ell\}$
\end{algorithmic}
\end{algorithm}

\subsection{Greedy queries}
The routing algorithm (Algorithm~\ref{alg:query}) uses only the edges to out-neighbors. It first reaches a local
minimum on the current layer, then reuses that vertex as the entry point on
the next denser layer. Upper layers therefore provide long-range progress,
while lower layers refine the search.

To be more specific, given $q\in\T^d$, start at any vertex of $X_L$. At a current vertex $v$ in layer $\ell$, choose an out-neighbor $w$ minimizing $d(w,q)$. Thus $d(w,q)\le d(u,q)$ for every $u\in N_\ell^+(v)$. Move to $w$ only if $d(w,q)<d(v,q)$. If there is no out-neighbor, or none is strictly closer to $q$ than $v$, stop the search on this layer and, if $\ell>0$, use $v$ as the entry point on layer $\ell-1$. The vertex where the search stops on layer $\ell$ is called the terminal vertex and is denoted by $W_\ell$. The algorithm returns $W_0$. If $X_L=\varnothing$, the routing fails.

\begin{algorithm}[t]
\small
\caption{Greedy top-down search for query $q$}\label{alg:query}
\begin{algorithmic}[1]
\REQUIRE Query $q\in\T^d$; hierarchy $(X_\ell)_{\ell=0}^L$ and neighbor lists $N_\ell^+(v)$ from Algorithm~\ref{alg:build}; 
\IF{$X_L=\varnothing$}
 \RETURN failure
\ENDIF
\STATE Choose any $v\in X_L$
\FOR{$\ell=L,L-1,\ldots,0$}
 \WHILE{$\exists\,u\in N_\ell^+(v)$ with $d(u,q)<d(v,q)$}
  \STATE $v\gets\arg\min_{u\in N_\ell^+(v)}d(u,q)$
 \ENDWHILE
 \STATE $W_\ell\gets v$
\ENDFOR
\RETURN $W_0$
\end{algorithmic}
\end{algorithm}

Let $T_\ell$ be the total number of greedy hops at layer $\ell$ and set
\begin{equation}
 T_{\rm tot}=\sum_{\ell=0}^L T_\ell,\qquad
 N_{\rm tot}=\sum_{\ell=0}^L|X_\ell|,\qquad
 A_{\rm tot}=\sum_{\ell=0}^L\sum_{v\in X_\ell}|N_\ell^+(v)|.
 \label{eq:resources}
\end{equation}
If $X_L=\varnothing$, set $T_\ell=0$ for all $\ell$. Algorithm~\ref{alg:query} uses
$O(K(T_{\rm tot}+L+1))$ distance queries, where $K$ is the number of sectors.

\section{Main results}
\label{sec:results}

\subsection{Class of navigable hierarchies}
We now define the deterministic coverage class of hierarchies that ensure greedy search with $(1+\eps)$ accuracy.

\begin{definition}[Coverage class]
\label{def:class}
Let $r_{\rm cov}>0$ and let $h_\ell$ be as in~\eqref{eq:scale}. Let $\mathfrak C_n$ consist of the finite nested configurations
$(X_0,\ldots,X_L)$ for which, at every $v\in X_\ell$,
\begin{enumerate}
\item at least $m=\lceil C_{\rm pool}\log n\rceil$ points of $X_\ell\setminus\{v\}$ lie in $B(v,r_{\rm cov}h_\ell)$;
\item $P_\ell(v)\cap\cC_j(v)\neq \emptyset$ for all $j$.
\end{enumerate}
\end{definition}
Define $R_\ell(q):=\min_{x\in X_\ell}d(x,q)$, with
$R_\ell(q):=\infty$ when $X_\ell=\varnothing$. The next theorem shows that, whenever an HNSW-inspired hierarchy belongs to the coverage class $\mathfrak C_n$, greedy search terminates close to the nearest point to the query.

\begin{maintheorem}[Deterministic $(1+\eps)$ accuracy]
\label{thm:deterministic}
Fix $d\ge1$ and $n>1$. For every $\eps>0$, there exists a sector diameter threshold $\omega_\varepsilon<\pi/3$ such that, for every hierarchy in $\mathfrak C_n$ with $\omega<\omega_\varepsilon$, the terminal vertex satisfies
\begin{equation}
 d(W_\ell,q)\le\min\{Bh_\ell,(1+\eps)R_\ell(q)\}.
 \label{eq:main-ann}
\end{equation}
Here, $B>r_{\rm cov}/(2\cos\omega)$ is any constant satisfying $Bh_L<1/4$. The bound holds simultaneously for every $q\in\T^d$, every nonempty layer $\ell$, and every starting vertex in that layer.
Furthermore, if $q\in X_\ell$, the layer search reaches $q$ exactly.
\end{maintheorem}

Geometrically distributed heights are standard in HNSW~\cite{MalkovYashunin2020}. However, Theorem~\ref{thm:deterministic} also applies to hierarchies generated by more general height distributions, provided that the conditions of Definition~\ref{def:class} hold.

\subsection{Accuracy for random point processes}
We now consider three classes of random inputs, for which we will show that Algorithm~\ref{alg:query} is successful with high probability. Poisson processes are the standard benchmark for point processes, and give a model where all points are independent and uniformly distributed. Determinantal point processes introduce local repulsion, while Cox processes allow clustering through a random intensity field; see Figure~\ref{fig:input-families}. Thus, these classes support wide ranges of point processes, allowing us to show that our navigation argument holds in a wide variety of settings.
We define these briefly below.
For background on Poisson and Cox processes, see
\cite{LastPenrose2017}; for determinantal point processes,
see \cite{Hough2006}.

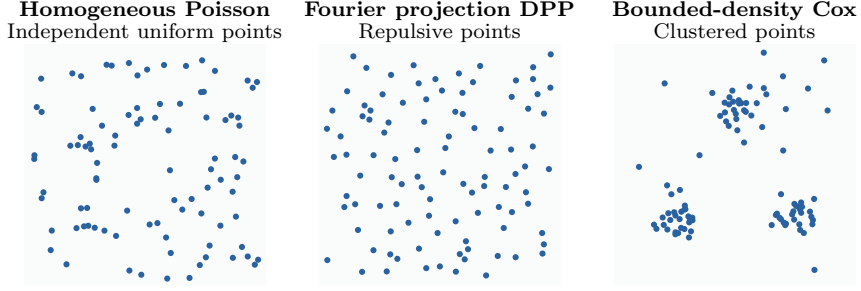
\begin{figure}[t]
\centering
\begin{tikzpicture}[x=1cm,y=1cm,font=\footnotesize]
\begin{scope}[xshift=0.0cm]
\node[black,font=\scriptsize\bfseries] at (1.5,3.55) {Homogeneous Poisson};
\node[black,font=\scriptsize] at (1.5,3.25) {Independent uniform points};
\draw[fill=hsnwGray!3,draw=none,line width=.3pt] (-.1,-.1) rectangle (3.1,3.1);
\fill[hsnwBlue] (0.041758,1.631900) circle (1.15pt);
\fill[hsnwBlue] (0.516674,1.758272) circle (1.15pt);
\fill[hsnwBlue] (0.072282,2.268881) circle (1.15pt);
\fill[hsnwBlue] (0.935284,2.023033) circle (1.15pt);
\fill[hsnwBlue] (2.805072,1.549593) circle (1.15pt);
\fill[hsnwBlue] (2.927656,2.524900) circle (1.15pt);
\fill[hsnwBlue] (1.286537,2.164746) circle (1.15pt);
\fill[hsnwBlue] (0.773539,1.785168) circle (1.15pt);
\fill[hsnwBlue] (2.376776,1.447268) circle (1.15pt);
\fill[hsnwBlue] (2.120366,1.054056) circle (1.15pt);
\fill[hsnwBlue] (2.315908,2.131338) circle (1.15pt);
\fill[hsnwBlue] (1.063007,2.817531) circle (1.15pt);
\fill[hsnwBlue] (1.755367,2.825267) circle (1.15pt);
\fill[hsnwBlue] (2.045921,2.813375) circle (1.15pt);
\fill[hsnwBlue] (1.565139,0.716377) circle (1.15pt);
\fill[hsnwBlue] (2.331514,1.203747) circle (1.15pt);
\fill[hsnwBlue] (2.786209,0.419563) circle (1.15pt);
\fill[hsnwBlue] (2.308799,0.282771) circle (1.15pt);
\fill[hsnwBlue] (1.668151,0.734698) circle (1.15pt);
\fill[hsnwBlue] (2.281544,2.471106) circle (1.15pt);
\fill[hsnwBlue] (0.811723,2.245271) circle (1.15pt);
\fill[hsnwBlue] (1.109922,1.892215) circle (1.15pt);
\fill[hsnwBlue] (0.160886,1.142119) circle (1.15pt);
\fill[hsnwBlue] (0.967650,0.635392) circle (1.15pt);
\fill[hsnwBlue] (0.870443,1.531009) circle (1.15pt);
\fill[hsnwBlue] (2.567839,2.097735) circle (1.15pt);
\fill[hsnwBlue] (2.259089,2.887376) circle (1.15pt);
\fill[hsnwBlue] (1.797915,0.001416) circle (1.15pt);
\fill[hsnwBlue] (2.358192,0.437945) circle (1.15pt);
\fill[hsnwBlue] (0.860874,1.307042) circle (1.15pt);
\fill[hsnwBlue] (2.500802,0.701200) circle (1.15pt);
\fill[hsnwBlue] (2.999672,0.250871) circle (1.15pt);
\fill[hsnwBlue] (0.838617,0.673860) circle (1.15pt);
\fill[hsnwBlue] (2.665439,0.276299) circle (1.15pt);
\fill[hsnwBlue] (1.099898,1.759340) circle (1.15pt);
\fill[hsnwBlue] (1.291166,2.851181) circle (1.15pt);
\fill[hsnwBlue] (0.598009,0.668797) circle (1.15pt);
\fill[hsnwBlue] (2.428858,2.747614) circle (1.15pt);
\fill[hsnwBlue] (2.735817,2.130390) circle (1.15pt);
\fill[hsnwBlue] (1.637672,2.127942) circle (1.15pt);
\fill[hsnwBlue] (2.092526,0.010102) circle (1.15pt);
\fill[hsnwBlue] (0.130492,2.645117) circle (1.15pt);
\fill[hsnwBlue] (1.297509,0.107775) circle (1.15pt);
\fill[hsnwBlue] (2.245051,2.467247) circle (1.15pt);
\fill[hsnwBlue] (0.651867,1.872120) circle (1.15pt);
\fill[hsnwBlue] (2.955757,0.185430) circle (1.15pt);
\fill[hsnwBlue] (0.256696,0.378063) circle (1.15pt);
\fill[hsnwBlue] (0.697259,0.691808) circle (1.15pt);
\fill[hsnwBlue] (1.782599,0.544646) circle (1.15pt);
\fill[hsnwBlue] (1.121083,0.565403) circle (1.15pt);
\fill[hsnwBlue] (1.267545,0.900071) circle (1.15pt);
\fill[hsnwBlue] (2.723832,1.342557) circle (1.15pt);
\fill[hsnwBlue] (0.753146,0.699846) circle (1.15pt);
\fill[hsnwBlue] (0.657829,2.677189) circle (1.15pt);
\fill[hsnwBlue] (2.662283,1.149885) circle (1.15pt);
\fill[hsnwBlue] (2.644966,2.157162) circle (1.15pt);
\fill[hsnwBlue] (0.615612,1.766302) circle (1.15pt);
\fill[hsnwBlue] (2.752446,2.023451) circle (1.15pt);
\fill[hsnwBlue] (0.738382,0.935259) circle (1.15pt);
\fill[hsnwBlue] (1.815012,1.960428) circle (1.15pt);
\fill[hsnwBlue] (0.038984,1.582941) circle (1.15pt);
\fill[hsnwBlue] (0.995920,2.828991) circle (1.15pt);
\fill[hsnwBlue] (2.706176,0.820977) circle (1.15pt);
\fill[hsnwBlue] (1.844322,0.339513) circle (1.15pt);
\fill[hsnwBlue] (1.399316,2.256756) circle (1.15pt);
\fill[hsnwBlue] (2.428112,1.588660) circle (1.15pt);
\fill[hsnwBlue] (0.543016,2.714513) circle (1.15pt);
\fill[hsnwBlue] (0.136592,2.205182) circle (1.15pt);
\fill[hsnwBlue] (1.455083,0.052091) circle (1.15pt);
\fill[hsnwBlue] (2.505422,1.308335) circle (1.15pt);
\fill[hsnwBlue] (2.862535,0.293691) circle (1.15pt);
\fill[hsnwBlue] (0.791200,1.710869) circle (1.15pt);
\fill[hsnwBlue] (0.656943,0.930334) circle (1.15pt);
\fill[hsnwBlue] (0.859806,1.328153) circle (1.15pt);
\fill[hsnwBlue] (2.249120,0.862579) circle (1.15pt);
\fill[hsnwBlue] (2.696414,1.608619) circle (1.15pt);
\fill[hsnwBlue] (1.833585,1.434029) circle (1.15pt);
\fill[hsnwBlue] (1.531877,2.765134) circle (1.15pt);
\fill[hsnwBlue] (1.446468,2.107873) circle (1.15pt);
\fill[hsnwBlue] (1.996831,0.913625) circle (1.15pt);
\fill[hsnwBlue] (1.966449,2.313138) circle (1.15pt);
\fill[hsnwBlue] (0.137271,1.064218) circle (1.15pt);
\fill[hsnwBlue] (0.717656,1.747563) circle (1.15pt);
\fill[hsnwBlue] (2.860667,2.619901) circle (1.15pt);
\fill[hsnwBlue] (2.228509,0.179651) circle (1.15pt);
\fill[hsnwBlue] (1.891449,1.039222) circle (1.15pt);
\fill[hsnwBlue] (0.465598,0.189834) circle (1.15pt);
\fill[hsnwBlue] (2.984695,2.600428) circle (1.15pt);
\fill[hsnwBlue] (0.271208,0.636205) circle (1.15pt);
\fill[hsnwBlue] (1.701109,2.310165) circle (1.15pt);
\fill[hsnwBlue] (1.399212,2.049498) circle (1.15pt);
\fill[hsnwBlue] (1.404358,2.687667) circle (1.15pt);
\end{scope}
\begin{scope}[xshift=3.9cm]
\node[black,font=\scriptsize\bfseries] at (1.5,3.55) {Fourier projection DPP};
\node[black,font=\scriptsize] at (1.5,3.25) {Repulsive points};
\draw[fill=hsnwGray!3,draw=none,line width=.3pt] (-.1,-.1) rectangle (3.1,3.1);
\fill[hsnwBlue] (0.767902,1.931774) circle (1.15pt);
\fill[hsnwBlue] (1.813622,1.253838) circle (1.15pt);
\fill[hsnwBlue] (0.474586,0.090445) circle (1.15pt);
\fill[hsnwBlue] (1.227085,0.397344) circle (1.15pt);
\fill[hsnwBlue] (2.964540,2.526122) circle (1.15pt);
\fill[hsnwBlue] (1.875080,0.389790) circle (1.15pt);
\fill[hsnwBlue] (1.322935,0.139361) circle (1.15pt);
\fill[hsnwBlue] (0.190429,1.881598) circle (1.15pt);
\fill[hsnwBlue] (0.941060,0.262630) circle (1.15pt);
\fill[hsnwBlue] (2.878660,0.454054) circle (1.15pt);
\fill[hsnwBlue] (0.014839,0.245584) circle (1.15pt);
\fill[hsnwBlue] (1.322475,2.371382) circle (1.15pt);
\fill[hsnwBlue] (0.585983,2.223927) circle (1.15pt);
\fill[hsnwBlue] (0.928281,0.985155) circle (1.15pt);
\fill[hsnwBlue] (0.420136,0.644535) circle (1.15pt);
\fill[hsnwBlue] (1.503404,0.162092) circle (1.15pt);
\fill[hsnwBlue] (2.453040,2.404241) circle (1.15pt);
\fill[hsnwBlue] (0.800277,0.096049) circle (1.15pt);
\fill[hsnwBlue] (1.831431,0.255804) circle (1.15pt);
\fill[hsnwBlue] (0.359428,2.940073) circle (1.15pt);
\fill[hsnwBlue] (2.376113,1.212211) circle (1.15pt);
\fill[hsnwBlue] (1.711881,0.758926) circle (1.15pt);
\fill[hsnwBlue] (0.857342,2.694723) circle (1.15pt);
\fill[hsnwBlue] (1.296056,2.074732) circle (1.15pt);
\fill[hsnwBlue] (2.503030,1.384413) circle (1.15pt);
\fill[hsnwBlue] (0.745530,2.437801) circle (1.15pt);
\fill[hsnwBlue] (0.269638,1.092924) circle (1.15pt);
\fill[hsnwBlue] (2.268456,0.300332) circle (1.15pt);
\fill[hsnwBlue] (0.480686,2.148002) circle (1.15pt);
\fill[hsnwBlue] (0.507398,2.290494) circle (1.15pt);
\fill[hsnwBlue] (2.173916,1.694362) circle (1.15pt);
\fill[hsnwBlue] (2.088728,1.224798) circle (1.15pt);
\fill[hsnwBlue] (2.006175,0.870607) circle (1.15pt);
\fill[hsnwBlue] (1.941173,2.598729) circle (1.15pt);
\fill[hsnwBlue] (2.096987,1.345686) circle (1.15pt);
\fill[hsnwBlue] (1.847751,0.494913) circle (1.15pt);
\fill[hsnwBlue] (0.276752,2.477976) circle (1.15pt);
\fill[hsnwBlue] (1.013726,2.326571) circle (1.15pt);
\fill[hsnwBlue] (0.824076,2.064839) circle (1.15pt);
\fill[hsnwBlue] (2.855687,0.603811) circle (1.15pt);
\fill[hsnwBlue] (0.383077,0.984645) circle (1.15pt);
\fill[hsnwBlue] (1.622010,1.880472) circle (1.15pt);
\fill[hsnwBlue] (0.571353,2.103267) circle (1.15pt);
\fill[hsnwBlue] (1.472819,0.603508) circle (1.15pt);
\fill[hsnwBlue] (1.530697,2.487762) circle (1.15pt);
\fill[hsnwBlue] (2.734952,2.639252) circle (1.15pt);
\fill[hsnwBlue] (0.846929,0.321436) circle (1.15pt);
\fill[hsnwBlue] (0.853151,0.625467) circle (1.15pt);
\fill[hsnwBlue] (2.035834,2.812863) circle (1.15pt);
\fill[hsnwBlue] (2.309093,1.979030) circle (1.15pt);
\fill[hsnwBlue] (2.509329,2.621924) circle (1.15pt);
\fill[hsnwBlue] (1.547261,1.406437) circle (1.15pt);
\fill[hsnwBlue] (2.456417,0.880946) circle (1.15pt);
\fill[hsnwBlue] (2.114729,0.042909) circle (1.15pt);
\fill[hsnwBlue] (0.913213,1.581838) circle (1.15pt);
\fill[hsnwBlue] (2.670834,2.247714) circle (1.15pt);
\fill[hsnwBlue] (0.711656,0.962727) circle (1.15pt);
\fill[hsnwBlue] (1.707357,2.207339) circle (1.15pt);
\fill[hsnwBlue] (2.669595,0.991099) circle (1.15pt);
\fill[hsnwBlue] (2.969272,2.966693) circle (1.15pt);
\fill[hsnwBlue] (2.879740,2.183401) circle (1.15pt);
\fill[hsnwBlue] (1.575636,1.727790) circle (1.15pt);
\fill[hsnwBlue] (0.101438,1.609757) circle (1.15pt);
\fill[hsnwBlue] (0.624740,0.376456) circle (1.15pt);
\fill[hsnwBlue] (0.531093,2.787575) circle (1.15pt);
\fill[hsnwBlue] (1.373571,1.050727) circle (1.15pt);
\fill[hsnwBlue] (2.731958,1.748547) circle (1.15pt);
\fill[hsnwBlue] (2.283065,0.947897) circle (1.15pt);
\fill[hsnwBlue] (1.202237,1.700043) circle (1.15pt);
\fill[hsnwBlue] (1.754429,2.714194) circle (1.15pt);
\fill[hsnwBlue] (1.766629,0.325340) circle (1.15pt);
\fill[hsnwBlue] (2.794901,2.439965) circle (1.15pt);
\fill[hsnwBlue] (1.714817,0.092868) circle (1.15pt);
\fill[hsnwBlue] (1.087948,0.740029) circle (1.15pt);
\fill[hsnwBlue] (1.208460,1.348869) circle (1.15pt);
\fill[hsnwBlue] (1.263754,1.210092) circle (1.15pt);
\fill[hsnwBlue] (0.449755,1.362934) circle (1.15pt);
\fill[hsnwBlue] (1.222606,2.583091) circle (1.15pt);
\fill[hsnwBlue] (1.359765,0.832281) circle (1.15pt);
\fill[hsnwBlue] (0.010037,1.982393) circle (1.15pt);
\fill[hsnwBlue] (1.444238,2.639303) circle (1.15pt);
\fill[hsnwBlue] (1.805568,2.177299) circle (1.15pt);
\fill[hsnwBlue] (2.348091,0.670232) circle (1.15pt);
\fill[hsnwBlue] (2.505002,0.284133) circle (1.15pt);
\fill[hsnwBlue] (2.736410,2.796700) circle (1.15pt);
\fill[hsnwBlue] (2.945176,1.449363) circle (1.15pt);
\fill[hsnwBlue] (1.960305,1.991407) circle (1.15pt);
\fill[hsnwBlue] (0.081866,2.223252) circle (1.15pt);
\fill[hsnwBlue] (0.625986,1.633067) circle (1.15pt);
\fill[hsnwBlue] (2.444521,1.707942) circle (1.15pt);
\fill[hsnwBlue] (2.630245,1.484001) circle (1.15pt);
\fill[hsnwBlue] (2.641573,1.143010) circle (1.15pt);
\fill[hsnwBlue] (0.267688,1.639422) circle (1.15pt);
\fill[hsnwBlue] (0.998677,1.250328) circle (1.15pt);
\fill[hsnwBlue] (2.860104,0.332470) circle (1.15pt);
\fill[hsnwBlue] (0.866817,1.342033) circle (1.15pt);
\fill[hsnwBlue] (1.155610,2.822448) circle (1.15pt);
\fill[hsnwBlue] (1.801856,2.457295) circle (1.15pt);
\fill[hsnwBlue] (1.928179,1.138716) circle (1.15pt);
\fill[hsnwBlue] (0.240236,0.947477) circle (1.15pt);
\end{scope}
\begin{scope}[xshift=7.8cm]
\node[black,font=\scriptsize\bfseries] at (1.5,3.55) {Bounded-density Cox};
\node[black,font=\scriptsize] at (1.5,3.25) {Clustered points};
\draw[fill=hsnwGray!3,draw=none,line width=.3pt] (-.1,-.1) rectangle (3.1,3.1);
\fill[hsnwBlue] (1.471009,2.486555) circle (1.15pt);
\fill[hsnwBlue] (0.889696,0.824027) circle (1.15pt);
\fill[hsnwBlue] (0.777390,1.173319) circle (1.15pt);
\fill[hsnwBlue] (0.941424,0.806857) circle (1.15pt);
\fill[hsnwBlue] (1.991007,0.708935) circle (1.15pt);
\fill[hsnwBlue] (2.179423,0.726869) circle (1.15pt);
\fill[hsnwBlue] (2.261367,0.849525) circle (1.15pt);
\fill[hsnwBlue] (2.114529,1.117225) circle (1.15pt);
\fill[hsnwBlue] (2.093662,0.806741) circle (1.15pt);
\fill[hsnwBlue] (1.556485,2.359610) circle (1.15pt);
\fill[hsnwBlue] (1.721676,2.157011) circle (1.15pt);
\fill[hsnwBlue] (2.367213,0.644322) circle (1.15pt);
\fill[hsnwBlue] (2.380178,1.007868) circle (1.15pt);
\fill[hsnwBlue] (0.773810,0.787657) circle (1.15pt);
\fill[hsnwBlue] (2.178238,2.242855) circle (1.15pt);
\fill[hsnwBlue] (1.887062,2.519704) circle (1.15pt);
\fill[hsnwBlue] (0.753917,0.602668) circle (1.15pt);
\fill[hsnwBlue] (1.865160,2.004774) circle (1.15pt);
\fill[hsnwBlue] (0.434665,0.714071) circle (1.15pt);
\fill[hsnwBlue] (1.625845,1.800698) circle (1.15pt);
\fill[hsnwBlue] (2.046623,0.842574) circle (1.15pt);
\fill[hsnwBlue] (1.589327,2.321520) circle (1.15pt);
\fill[hsnwBlue] (0.897178,0.544916) circle (1.15pt);
\fill[hsnwBlue] (2.345732,0.947744) circle (1.15pt);
\fill[hsnwBlue] (2.730324,2.223978) circle (1.15pt);
\fill[hsnwBlue] (0.549321,0.809086) circle (1.15pt);
\fill[hsnwBlue] (2.113226,0.787129) circle (1.15pt);
\fill[hsnwBlue] (2.515824,0.818043) circle (1.15pt);
\fill[hsnwBlue] (1.310037,2.152671) circle (1.15pt);
\fill[hsnwBlue] (1.388341,2.079403) circle (1.15pt);
\fill[hsnwBlue] (1.369283,2.327770) circle (1.15pt);
\fill[hsnwBlue] (2.149734,0.742197) circle (1.15pt);
\fill[hsnwBlue] (1.537034,1.962393) circle (1.15pt);
\fill[hsnwBlue] (2.392461,0.950620) circle (1.15pt);
\fill[hsnwBlue] (1.205731,2.455858) circle (1.15pt);
\fill[hsnwBlue] (1.886699,2.988507) circle (1.15pt);
\fill[hsnwBlue] (0.900050,0.745275) circle (1.15pt);
\fill[hsnwBlue] (2.556901,0.116205) circle (1.15pt);
\fill[hsnwBlue] (0.563328,0.711351) circle (1.15pt);
\fill[hsnwBlue] (1.373273,2.199086) circle (1.15pt);
\fill[hsnwBlue] (2.356505,0.718829) circle (1.15pt);
\fill[hsnwBlue] (1.036453,1.632639) circle (1.15pt);
\fill[hsnwBlue] (0.842864,0.966927) circle (1.15pt);
\fill[hsnwBlue] (0.837303,0.913815) circle (1.15pt);
\fill[hsnwBlue] (1.920270,2.438030) circle (1.15pt);
\fill[hsnwBlue] (1.426764,2.304139) circle (1.15pt);
\fill[hsnwBlue] (1.626241,2.226938) circle (1.15pt);
\fill[hsnwBlue] (2.380667,0.871342) circle (1.15pt);
\fill[hsnwBlue] (0.850275,0.919066) circle (1.15pt);
\fill[hsnwBlue] (0.717932,0.587427) circle (1.15pt);
\fill[hsnwBlue] (0.943890,0.772440) circle (1.15pt);
\fill[hsnwBlue] (0.611513,0.650223) circle (1.15pt);
\fill[hsnwBlue] (1.793641,2.320473) circle (1.15pt);
\fill[hsnwBlue] (0.623820,0.865697) circle (1.15pt);
\fill[hsnwBlue] (1.523357,2.111585) circle (1.15pt);
\fill[hsnwBlue] (1.481332,2.340215) circle (1.15pt);
\fill[hsnwBlue] (0.923879,0.626013) circle (1.15pt);
\fill[hsnwBlue] (0.664707,1.116444) circle (1.15pt);
\fill[hsnwBlue] (1.560494,2.315911) circle (1.15pt);
\fill[hsnwBlue] (2.497428,0.886138) circle (1.15pt);
\fill[hsnwBlue] (2.277705,0.889994) circle (1.15pt);
\fill[hsnwBlue] (1.423205,2.020577) circle (1.15pt);
\fill[hsnwBlue] (2.170449,0.708198) circle (1.15pt);
\fill[hsnwBlue] (0.793033,0.671886) circle (1.15pt);
\fill[hsnwBlue] (0.836556,0.656572) circle (1.15pt);
\fill[hsnwBlue] (1.648776,2.206221) circle (1.15pt);
\fill[hsnwBlue] (2.471137,2.514888) circle (1.15pt);
\fill[hsnwBlue] (0.718131,0.835385) circle (1.15pt);
\fill[hsnwBlue] (2.366549,0.913699) circle (1.15pt);
\fill[hsnwBlue] (2.415995,0.603552) circle (1.15pt);
\fill[hsnwBlue] (0.719563,0.613238) circle (1.15pt);
\fill[hsnwBlue] (0.796027,0.885735) circle (1.15pt);
\fill[hsnwBlue] (2.140483,2.676441) circle (1.15pt);
\fill[hsnwBlue] (2.291589,0.988933) circle (1.15pt);
\fill[hsnwBlue] (0.608123,1.234122) circle (1.15pt);
\fill[hsnwBlue] (2.143174,1.694769) circle (1.15pt);
\fill[hsnwBlue] (2.060879,0.860266) circle (1.15pt);
\fill[hsnwBlue] (2.536017,0.775399) circle (1.15pt);
\fill[hsnwBlue] (1.516173,2.178816) circle (1.15pt);
\fill[hsnwBlue] (1.699554,2.539166) circle (1.15pt);
\fill[hsnwBlue] (2.685085,2.886646) circle (1.15pt);
\fill[hsnwBlue] (0.701456,0.878426) circle (1.15pt);
\fill[hsnwBlue] (1.653044,2.323017) circle (1.15pt);
\fill[hsnwBlue] (2.315662,0.960638) circle (1.15pt);
\fill[hsnwBlue] (1.491519,2.231361) circle (1.15pt);
\fill[hsnwBlue] (0.236441,1.478424) circle (1.15pt);
\fill[hsnwBlue] (0.575707,2.582506) circle (1.15pt);
\fill[hsnwBlue] (0.910988,0.774753) circle (1.15pt);
\fill[hsnwBlue] (2.529207,0.719555) circle (1.15pt);
\fill[hsnwBlue] (0.649918,0.775791) circle (1.15pt);
\fill[hsnwBlue] (2.288244,0.781244) circle (1.15pt);
\fill[hsnwBlue] (2.552015,1.054303) circle (1.15pt);
\fill[hsnwBlue] (0.467236,0.659198) circle (1.15pt);
\fill[hsnwBlue] (1.432476,2.398917) circle (1.15pt);
\end{scope}
\end{tikzpicture}
\caption{Point processes on $\mathbb T^2$, with common expected count $n = 100$. Left: homogeneous Poisson intensity $n$. Middle: the Fourier projection DPP. Right: Cox with intensity having three modes at three different points.} 
\Description{Simulated independent Poisson points, repulsive Fourier determinantal points, and a bounded-density Cox process with three clusters near the vertices of a triangle, each model having expected count 100.}
\label{fig:input-families}
\end{figure}

A random input $X_0$ is a finite point process on $\T^d$. 
We describe the three random input models through their
joint intensities.

\begin{definition}[Joint intensities]
\label{def:inputs}
Let $X_0$ be a simple point process on $\T^d$. Its joint intensities with respect to
Lebesgue measure are nonnegative measurable functions
$\varphi^{(k)}:(\T^d)^k\to[0,\infty)$, $k\ge1$, such that, for every
collection of pairwise disjoint Borel sets
$A_1,\ldots,A_k\subseteq\T^d$,
\[
\mathbb E\Big[\prod_{i=1}^k |X_0\cap A_i|\Big] =\int_{A_1\times\cdots\times A_k} \varphi^{(k)}(x_1,\ldots,x_k)\,dx_1\cdots dx_k.
\]

The three input families have the following joint
intensities. The constants $0<a\le b<\infty$ below are
independent of $n$ and $d$.

\begin{enumerate}
\item \emph{Poisson input.}
A homogeneous Poisson process with rate $\lambda\in (0,\infty)$ has
$\varphi^{(k)}(x_1,\ldots,x_k)=\lambda^k$.
We take $\lambda=n$, giving $n$ points in expectation
on the unit-volume torus. Under this process, the numbers of points in disjoint sets $A_i$ are independent.

\item \emph{Determinantal input.
A Hermitian determinantal point process (DPP) has}
\[
\varphi^{(k)}(x_1,\ldots,x_k) =\det[\cK_n(x_i,x_j)]_{i,j=1}^k,
\]
where the kernel $\cK_n$ is measurable and Hermitian (i.e., $\cK_n(x,y) = \overline{\cK_n(y,x)}$). We assume that $an\le \cK_n(x,x)\le bn$ almost everywhere and that $\cK_n(x,x)=\sum_i\alpha_{n,i}|\phi_{n,i}(x)|^2$ for almost every $x\in\T^d,$ where $\alpha_{n,i}$ and $\phi_{n,i}$ are the eigenvalues and orthonormal eigenfunctions of the associated integral operator. Additionally, we assume that $0\le\alpha_{n,i}\le 1$ for every $i$ and $\sum_i\alpha_{n,i}<\infty$.

\item \emph{Cox input.}
The Cox process has 
\[
\varphi^{(k)}(x_1,\ldots,x_k)=n^k\mathbb E\Big[\prod_{i=1}^kF_n(x_i)\Big],
\]
where $F_n$ is a jointly measurable random function. 
Equivalently, conditional on $F_n$, $X_0$ is a Poisson process with inhomogeneous intensity $nF_n(x)$. We assume $a\le F_n(x)\le b$ almost everywhere, almost surely. We refer to the random function $F_n$ as the environment.
\end{enumerate}
\end{definition}

Under these sampling distributions, we will consider the events $\cG_n=\{(X_\ell)_{\ell\le L}\in
\mathfrak C_n\}$ and
$\cH_n=\cG_n \cap\{X_L\ne\varnothing\}$.
\noindent The first property of the coverage class $\mathfrak C_n$ in Definition~\ref{def:class} ensures that the greedy arcs on each layer are not too long.
The second property ensures that the graph contains arcs toward every directional sector $\cC_j(v)$. By construction, in $\mathfrak C_n$, every selected sector neighbor is nearest in that sector and lies within $r_{\rm cov}h_\ell$ of its source.
The next theorem shows that each of the three input models yields a randomly sampled hierarchy satisfying the directional coverage condition with high probability, and hence yields a uniform $(1+\varepsilon)$-approximation guarantee for greedy search.

\begin{maintheorem}[Accuracy for Poisson, determinantal and Cox input]
\label{thm:stochastic}
Fix $\eps>0$, a retention probability $\rho\in(0,1)$,
and a failure exponent $M>2$.
Let $d=o(\log n/\log\log n)$.
Suppose that $X_0$ is one of the three point process inputs
described above, with the stated intensity bounds.
There exist sector partitions and dimension-dependent
parameters with
\[
\begin{gathered}
\omega<\omega_\varepsilon,\quad
K=\exp(O(d)), \quad 
C_{\rm pool}=\exp(O(d)),\quad r_{\rm cov},B=O(\sqrt d),\quad
D=\exp\!\left(\tfrac12 d\log d+O(d)\right),
\end{gathered}
\]
such that, for sufficiently large $n$,
\begin{equation}
\PP(\cG_n^c)=o(n^{-M}),\qquad
\PP(\cH_n^c)=O(n^{-M}).
\label{eq:good-probability}
\end{equation}
Thus, with probability $1-O(n^{-M})$, the approximation guarantees of Theorem~\ref{thm:deterministic} hold.
Explicit parameter choices are given in Appendix~\ref{app:dimension}.
\end{maintheorem}

\subsection{Number of hops}
The next theorem gives bounds on the total number of hops as defined in \eqref{eq:resources} for all three point process input classes.

\begin{maintheorem}[Hop count]
\label{thm:complexity}
Let $d=o(\log n/\log\log n)$ and assume the hypotheses of Theorem~\ref{thm:stochastic}. For every fixed query, or a random query independent of the marked input,
\begin{equation}
 \E T_{\rm tot}=O\bigl(\exp(\tfrac12 d\log d+O(d))\log n\bigr).
 \label{eq:cost-main}
\end{equation}
\end{maintheorem}

For fixed dimension, Theorem~\ref{thm:complexity} yields a logarithmic expected hop count. When the dimension is allowed to grow slowly with \(n\), the dependence is captured by the explicit factor \(\exp(\tfrac12 d\log d + O(d))\), which is $n^{o(1)}$ in the stated dimension regime.

\section{Strategy of the proofs}
\label{sec:roadmap} 
The following three items explain the main ideas behind the proofs.

\paragraph{Terminal vertices must be close.} Under the conditions of Definition~\ref{def:class}, if a vertex is far from the query, then one of its out-neighbors points approximately toward the query direction. A cosine-rule calculation (Lemma~\ref{lem:progress}) then shows that this neighbor strictly decreases the query distance. Thus, when greedy search stops, the terminal vertex cannot be far from the query.
 We use this argument to finally show that sufficiently narrow sectors yield a $(1+\varepsilon)$ approximation guarantee.

\paragraph{Search remains close to the query.} The $(1+\varepsilon)$ approximation~\eqref{eq:main-ann} shows that the search remains within a small neighborhood of the query on every non-top layer. Since each greedy step visits a new vertex, the number of hops is bounded by the number of points in that neighborhood. Theorem~\ref{thm:complexity} therefore reduces to bounding these point counts (Proposition~\ref{prop:layer-cost}).

\paragraph{Usable edge in every direction.} The deterministic theorem assumes that every sector $\mathcal{C}_j(v)$ has non-empty intersection with the candidate pool. To prove that this event occurs with high probability for the random point processes, we show that there exists a small ball containing at least one point in every sector and at most $m$ points with high probability. The probabilistic analysis then shows that this situation occurs simultaneously for all vertices and layers with high probability (Proposition~\ref{prop:coverage}).

\section{Numerical experiments}
\label{sec:experiments}

We experimentally evaluate two different questions: how close the returned point is to the
true nearest neighbor and how efficient the greedy route is compared to an exact shortest path. 
The first experiment uses the $d$-dimensional torus with varying $d$, the second one the unit two-dimensional torus. 
We use the following random processes:
\begin{enumerate}
\item A homogeneous Poisson point process of intensity $n$.
\item A DPP obtained by dividing the torus into $n$ equal-area
rectangles $C_1,\ldots,C_n$ and independently drawing one uniform point in each rectangle.
Its kernel is
\begin{equation*}
\mathcal K_n(x,y)=n\sum_{i=1}^n\1_{C_i}(x)\1_{C_i}(y).
\end{equation*} 
\item A Cox process with independent uniform phases
$\Phi_1,\Phi_2\in[0,2\pi)$ and conditional intensity $nF(x)$, where
\begin{equation*}
 F(x)=1+0.75\cos(2\pi x_1+\Phi_1)\cos(2\pi x_2+\Phi_2).
\end{equation*}
\end{enumerate}

\paragraph{Impact of dimension.}
To investigate the effect of the underlying dimension, we compare $K=16$ and $K=64$ sectors on the same
points and queries. Figure~\ref{fig:dimension-accuracy} shows that increasing the dimension while keeping the number of angular sectors $ K$ fixed gives a worse approximation ratio. For $K=16$ the mean approximation ratio remains within a 10\% error up to dimension 6, while for $K=64$, it remains within that range up to dimension 12. The \emph{recall}, the fraction of times the algorithm succeeds in finding the exact nearest neighbor, also decreases as the dimension increases. This can be mitigated by increasing $K$, at the cost of higher computational effort. 

\begin{figure}[t]
\centering
\includegraphics[width=\linewidth]{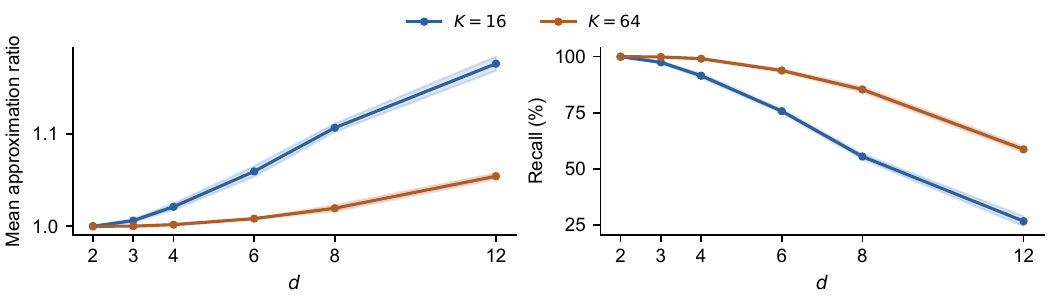}
\caption{Poisson input at intensity $20{,}000$: approximation ratio and exact
nearest-neighbor recall as dimension increases. Each plotted point shows the mean over five independent realizations, with $1{,}000$ queries per realization. The shaded areas form the 99\% intervals.}
\Description{For Poisson input in dimensions 2 through 12, mean approximation ratios rise and exact-neighbor recall falls. Sixty-four directions give better accuracy and require more distance evaluations than sixteen directions.}
\label{fig:dimension-accuracy}
\end{figure}

\paragraph{Shortest paths}
While our results show that the algorithm has a $(1+\varepsilon)$ approximation guarantee for nearest-neighbor search, we now evaluate how efficient the path the algorithm takes is compared to the shortest route. We therefore define the \emph{route stretch} as the hop count of the route returned by the algorithm divided by the hop count of the shortest route toward that target. 
Figure~\ref{fig:stretch_coverage} shows that the route stretch is typically at most 1.5, with some exceptions. Thus, while the algorithm approximately finds the closest point, the route it takes may be substantially longer than optimal in some cases. This is because in some cases, the exact shortest path
descends to a lower layer before reaching a greedy local minimum, whereas Algorithm~\ref{alg:query}
continues on its current layer as long as an improving edge is available.

\begin{figure}
    \centering
    \includegraphics[width=\linewidth]{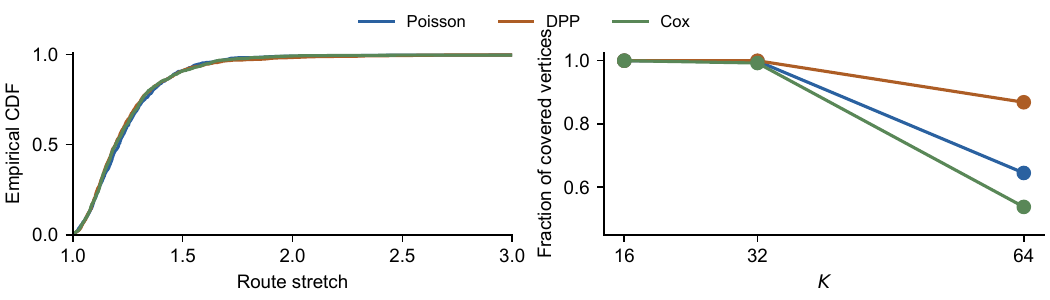}
    \caption{
    Left: route-stretch distributions at $K=32$. 
Right: fraction of vertices satisfying Condition~\ref{def:class}(2) against $K$ at $n=20{,}000$}
    \label{fig:stretch_coverage}
    \Description{Two figures. The left one plots the route stretch against the empirical CDF while the right one plots K against the fraction of covered vertices. Both figures compare the three input classes.}
\end{figure}

The right panel of Figure~\ref{fig:stretch_coverage} checks Condition~\ref{def:class}(2) and what fraction of vertices satisfy this condition for the families of stochastic point processes. While our results ensure that this fraction tends to one as $n\to\infty$, the figure shows that for finite $n$, this fraction drops quickly, especially for large $K$. Interestingly, how quickly this descends depends on the process types as well. For the DPP, which is more regular, 90\% of the vertices still satisfy the condition when $K=64$, whereas for the more clustered Cox process, only 50\% of the vertices still satisfy it.

\section{Proofs of Theorems~\ref{thm:deterministic} and~\ref{thm:stochastic}}
\label{sec:accuracy-proof}

For the proof of the deterministic result (Theorem~\ref{thm:deterministic}) in Section~\ref{sec:determ_proof}, we know that every candidate pool contains a representative from every sector. Under this assumption, we show that greedy search always has an improving move when it is far from the query. We then show that a terminal vertex must be close to the true nearest neighbor, yielding the deterministic $(1+\varepsilon)$ approximation. 

Section~\ref{subsec:criterion} then shows that the processes of Definition~\ref{def:inputs} satisfy the criteria for Theorem~\ref{thm:deterministic} with high probability. This then yields the accuracy guarantees of Theorem~\ref{thm:stochastic}.

\subsection{Directional coverage forces relative accuracy}\label{sec:determ_proof}

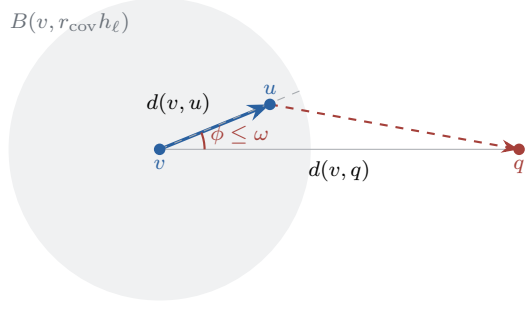
\begin{figure}[t]
\centering
\begin{tikzpicture}[
    scale=1.08,
    >=Stealth,
    every node/.style={font=\scriptsize}
]
\coordinate (v) at (0,0);
\coordinate (q) at (4.4,0);
\coordinate (u) at (1.35,0.55);
\fill[hsnwGray!10] (v) circle (1.85);
\draw[hsnwGray!55] (v) -- (q);
\draw[hsnwBlue,very thick,->] (v) -- (u);
\draw[hsnwRed,thick,dashed,->] (u) -- (q);
\draw[hsnwGray!60,dashed] (v) -- (1.8,0.75);
\draw[hsnwRed,thick] (0.55,0) arc[start angle=0,end angle=22,radius=0.55];

\fill[hsnwBlue] (v) circle (2pt) node[below] {$v$};
\fill[hsnwBlue] (u) circle (2pt) node[above] {$u$};
\fill[hsnwRed] (q) circle (2pt) node[below] {$q$};
\node[below] at (2.2,0) {$d(v,q)$};
\node[above left] at (0.75,0.3) {$d(v,u)$};
\node[hsnwRed] at (0.98,0.18) {$\phi\le\omega$};
\node[hsnwGray] at (-1.1,1.55) {$B(v,r_{\rm cov}h_\ell)$};
\end{tikzpicture}
\caption{From the current vertex $v$, the neighbor $u \in N_\ell^+(v)$ lies in the sector containing the direction toward the query $q$, so the angle between these directions satisfies
$\phi\le\omega$. Since the hierarchy belongs to $\mathfrak C_n$, the length of the edge $d(v,u)$ is at most $r_{\rm cov}h_\ell$.
Lemma~\ref{lem:progress} then shows that $u$ is closer to $q$ than
$v$ whenever $d(v,u)<2d(v,q)\cos\phi$.}
\label{fig:forward-progress}
\Description{A short edge in the query sector yields a decrease in distance by the cosine rule.}
\end{figure}

This subsection begins with Proposition~\ref{prop:geometry} which describes the core geometric property of the hierarchy model by giving upper bounds for the distance between the terminal vertex of a layer and the query.
For that, define the approximation factor
\begin{equation}\label{eq:cann}
    C_{\rm ann}:= (1-2\sin(\omega/2))^{-1}.
\end{equation}
Later, we will choose $\omega_\varepsilon$ so that $C_{\rm ann}< 1+\varepsilon$. We recall that $R_\ell(q)=\min_{x\in X_\ell}d(x,q)$.
After stating the proposition, we prove two lemmas quantifying the progress towards the query which we need for the proof of the proposition. We conclude with the proof of Theorem~\ref{thm:deterministic}.

\begin{proposition}[Distance of terminal vertex]
\label{prop:geometry}
Suppose $\omega<\pi/3$, $B>r_{\rm cov}/(2\cos\omega)$ and $Bh_L<1/4$.
For every hierarchy $(X_0,\ldots,X_L) \in \mathfrak C_n$, the following bounds hold simultaneously for every query $q\in \T^d$, every nonempty layer $\ell$, every starting vertex in that layer and every layerwise terminal vertex $W_\ell$:
\[
 d(W_\ell,q)\le Bh_\ell,
 \qquad d(W_\ell,q)\le C_{\rm ann}R_\ell(q).
\]
\end{proposition}

The following lemma states that a sufficiently short step pointing toward the query strictly decreases the distance to the query, see Figure~\ref{fig:forward-progress}. Recall that $\delta(x,y)$ is the canonical lift on the torus. 

\begin{lemma}[Progress]
\label{lem:progress}
Let $u,v,q\in\T^d$ be such that $\delta(q,v)$ and $\delta(u,v)$ are nonzero. \\
Let $\phi := \angle (\delta(q,v),\delta(u,v))$.
If $d(v,u)<2d(v,q)\cos\phi$, then 
\[
    d(u,q)<d(v,q).
\]
\end{lemma}
\begin{proof}
Since $\delta(q,v)-\delta(u,v)$ is a lift of $q-u$, the definition of the torus distance gives
\[
d(u,q)
\le
\|\delta(q,v)-\delta(u,v)\|.
\]
By squaring both sides of the previous inequality we obtain
\[
d(u,q)^2
\le
d(v,q)^2+d(v,u)^2-2d(v,q)d(v,u)\cos\phi.
\]
By using
\[
d(v,u)<2d(v,q)\cos\phi,
\]
we get
\[
d(u,q)^2
\le
d(v,q)^2+d(v,u)^2 -2d(v,q)d(v,u)\cos\phi<d(v,q)^2
\]
and hence
$d(u,q)<d(v,q)$.
\end{proof}

The next lemma provides conditions in~\eqref{eq:cone descent assumptions} under which a step toward a point near the query still makes progress. For an illustration, see Figure~\ref{fig:cone-descent}.

\begin{lemma}[Cone descent]
\label{lem:cone}
Let $x,y,z\in\mathbb R^d \setminus \{\mathbf{0}\}$. Suppose
\begin{align}\label{eq:cone descent assumptions}
(i)~  \|x-y\|<C_{\rm ann}^{-1}\|x\|,
\qquad
(ii)~  0<\|z\|\le\|y\|,
\qquad
(iii)~ \angle(y,z)\le\omega.
\end{align}
Then, $\|x-z\|<\|x\|$.
\end{lemma}
\begin{proof}
By scaling all three vectors by $\|x\|^{-1}$ and applying an orthogonal transformation, we may assume that $x$ is the first standard basis vector. Define $\psi := \angle(x,y)$ and $\phi := \angle(x,z)$. In an auxiliary plane, set
$x'=(1,0)$, $y'=\|y\|(\cos\psi,\sin\psi)$ and
$w=(\cos\omega,-\sin\omega)$. By construction, we have 
\begin{align*}
\angle(x',y') = \psi, \quad \|y-x\| = \|y' - x'\| \quad \text{ and } \quad \|x'-w \| = 2\sin(\frac{\omega}{2}) = 1- C_{\rm ann}^{-1}.
\end{align*}
Combining this with assumption $(i)$, we obtain
\[
 \|y'-w\| \le\|y'-x'\|+\|x'-w\| = \|y-x\|+1-C_{\rm ann}^{-1}< 1.
\]
This implies $$ \|y'-w\|^2 = \| y \|^2 +1 -2\| y \|\cos(\psi + \omega)< 1$$ which yields $$\|y\|<2\cos(\psi+\omega).$$ Since $\psi \le \pi$ and $\omega < \frac{\pi}{3}$, we have $0 \le \psi+\omega<\frac{4\pi}{3}$. Thus, $\cos(\psi+\omega)> 0$ implies $\psi+\omega <\frac{\pi}{2}$.
Applying assumption $(iii)$, we compute $\phi=\angle(x,z)\le \angle(x,y) + \angle(y,z) \le \psi+\omega$.
By assumption~$(ii)$ and the monotonicity of the cosine, we get
$$0< \|z\|\le \|y\|<2\cos(\psi+\omega)\le 2\cos\phi.$$ Therefore
$$ \|x-z\|^2=1+\|z\|^2-2\|z\|\cos\phi = 1+\|z\|(\|z\| - 2\cos\phi)<1,$$ which implies the statement.
\end{proof}

\begin{figure}[t]
\centering
\begin{tikzpicture}[
    scale=0.7,
    >=Stealth,
    every node/.style={font=\scriptsize}
]
\coordinate (W) at (0,0);\coordinate (q) at (2.8,0);
\coordinate (x) at (3.15,.42);\coordinate (u) at (1.95,.49);
\fill[hsnwBlue!7] (q) circle (2.8); 
\fill[hsnwRed!6] (W)--(18:3.8) arc (18:3:3.8)--cycle; 
\begin{scope} \clip (q) circle (2.8); \fill[purple!25] (W)--(18:3.8) arc (18:3:3.8)--cycle; \end{scope} 
\draw[hsnwGray] (W)--(q); 
\draw[hsnwGray,dashed] (W)--(x); 
\draw[hsnwBlue,very thick,->] (W)--(u); \draw[hsnwRed,thick] (q)--(x); \fill (W) circle (1.7pt) node[below left] {$W_\ell$}; 
\fill[hsnwRed] (q) circle (1.7pt) node[below] {$q$}; \fill (x) circle (1.7pt) node[above right] {$x_\ell$}; 
\fill[hsnwBlue] (u) circle (1.7pt) node[above] {$u$}; \node[below] at (1.25,0) {$d(W_\ell,q)$}; 
\node[hsnwBlue,align=center] at (2.8,2.2) {points closer to $q$\\than $W_\ell$}; 
\end{tikzpicture}
\caption{If the true nearest point $x_\ell$ were substantially closer to $q$ than the terminal vertex $W_\ell$, then the neighbor $u$ of $W_\ell$ in the same sector must also be closer to $q$ than $W_\ell$ (red sector). Such a neighbor would be an improving greedy move, so $W_\ell$ cannot be terminal. Lemma~\ref{lem:cone} makes this statement quantitative.}
\Description{A possible terminal point lies on a ball centered at the query. A nearest point close to the query forces the out-neighbor in its sector inside that ball.}
\label{fig:cone-descent}
\end{figure}
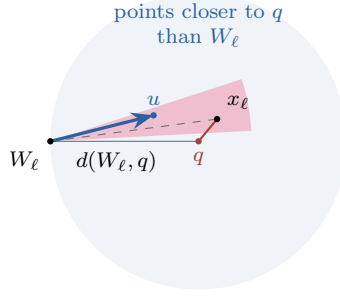

\begin{proof}[Proof of Proposition~\ref{prop:geometry}]
Fix a query $q\in \T^d$ and a nonempty layer $\ell \in \{0,\ldots, L\}$. 
We begin by showing $d(W_\ell,q)\le Bh_\ell$. Let $v \in X_\ell$ and suppose $d(v,q)>Bh_\ell$. Our goal is to show that no such vertex $v$ can be the terminal vertex. Since the hierarchy belongs to $\mathfrak C_n$, there is a neighbor $u\in N_\ell^+(v)$ in the sector containing $\delta(q,v)/\|\delta(q,v)\|$ that fulfills $d(v,u)\le r_{\rm cov}h_\ell$. Together with the assumption $B>r_{\rm cov}/(2\cos\omega)$ this yields
\[
 d(v,u)\le r_{\rm cov}h_\ell<2Bh_\ell\cos\omega <2d(v,q)\cos\omega.
\]
By construction of the sectors, we have $\phi = \angle(\delta(q,v),\delta(u,v))\le \omega < \frac{\pi}{3}$. Since the cosine is decreasing on $[0,\pi]$, we get
\[
 d(v,u)<2d(v,q)\cos\omega \le 2d(v,q)\cos\phi,
\]
so Lemma~\ref{lem:progress} gives
\[
    d(u,q) < d(v,q).
\]
Hence, no vertex $v$ with $d(v,q)>Bh_\ell$ can be terminal. Since each
greedy step strictly decreases the distance to $q$ and $X_\ell$ is
finite, the search terminates at a vertex $W_\ell$ satisfying $d(W_\ell,q)\le Bh_\ell$.

For the second bound, fix a terminal $W_\ell$ and a nearest layer point $x_\ell$. If $W_\ell=q$, there is nothing to prove. Since $R_\ell(q)\le d(W_\ell,q)\le Bh_\ell$ and $Bh_L<1/4$, we get
\[
 \|\delta(q,W_\ell)-\delta(q,x_\ell)\|\le d(W_\ell,q)+R_\ell(q)\le2Bh_\ell \le2Bh_L <\tfrac{1}{2}.
\]
This means that 
\begin{align*}
\delta(q,W_\ell)-\delta(q,x_\ell)=\delta(x_\ell,W_\ell) \quad \text{ and } \quad
\|\delta(q,W_\ell)-\delta(x_\ell,W_\ell)\|=R_\ell(q).
\end{align*}
If we assume $d(W_\ell,q)>C_{\rm ann}R_\ell(q)$, then $R_\ell(q)<C_{\rm ann}^{-1}\|\delta(q,W_\ell)\|$ and
$x_\ell\ne W_\ell$. Take $u\in N_\ell^+(W_\ell)$ as the neighbor in the sector of $\delta(x_\ell,W_\ell)$. Since by construction $u$ is the closest neighbor to $W_\ell$ in that sector, we obtain 
\begin{align*}
    0<d(W_\ell,u)\le d(W_\ell,x_\ell) \quad \text{ and } \quad \angle(\delta(x_\ell,W_\ell),\delta(u,W_\ell))\le\omega.
\end{align*}
Then, Lemma~\ref{lem:cone} yields
\[
 d(u,q)\le\|\delta(q,W_\ell)-\delta(u,W_\ell)\|<\|\delta(q,W_\ell)\|=d(W_\ell,q),
\]
which contradicts that $W_\ell$ is the terminal vertex. Thus, $d(W_\ell,q)\le C_{\rm ann}R_\ell(q)$.
\end{proof}

\begin{proof}[Proof of Theorem~\ref{thm:deterministic}]
Given $\eps>0$, set
\[
\omega_\varepsilon = 2\arcsin\!\Big(\frac{\eps}{2(1+\eps)}\Big).
\]
Then, for all $\omega < \omega_\varepsilon$ we obtain
\[
1-2\sin(\omega/2) > 1-2\sin(\omega_\varepsilon/2) = \frac{1}{1+\eps},
\]
so that by~\eqref{eq:cann} we have $C_{\rm ann} < 1+\eps.$
Proposition~\ref{prop:geometry} therefore yields
\begin{align}\label{eq:lastboundthm1}
    d(W_\ell,q) \le \min \{Bh_\ell,C_{\rm ann}R_\ell(q)\} \le\min \{Bh_\ell,(1+\eps)R_\ell(q)\}.
\end{align}
If $q\in X_\ell$, then $R_\ell(q) = 0$ and equation~\eqref{eq:lastboundthm1} implies exact recovery.
\end{proof}

\subsection{A probabilistic criterion for uniform coverage}
\label{subsec:criterion}
For Theorem~\ref{thm:stochastic}, we want to show that the point process satisfies the conditions of Definition~\ref{def:class} with high probability, so we need local count estimates for the other points conditional on a point being present at a given location. For this we use the reduced Palm law.
Write $$\cN_\ell(S):=|X_\ell\cap S|, \qquad \mu_\ell(S):=\int_S\lambda_\ell(x)\,dx,$$ where $\lambda_\ell$ is the intensity of $X_\ell$.

For $\lambda_\ell(v)>0$, let $\PP_\ell^{!v}$ denote the reduced Palm law at $v$, with expectation $\E_\ell^{!v}$; see e.g.,~\cite{CoeurjollyMollerWaagepetersen2017}. Informally, this is the law of the point configuration conditioned on containing $v$, after removing that point. It is characterized by the reduced Campbell identity~\cite{CoeurjollyMollerWaagepetersen2017}: for every nonnegative measurable $g$,
\begin{equation}
\E\sum_{v\in X_\ell}
g\bigl(v,X_\ell\setminus\{v\}\bigr)
=
\int_{\T^d}
\lambda_\ell(v)
\E_\ell^{!v}[g(v,X)]\,dv.
\label{eq:campbell}
\end{equation}
Here, $X$ denotes the reduced configuration under $\PP_\ell^{!v}$. Define $\mu_\ell^{!v}(S):=\E_\ell^{!v}|X\cap S|$. Recall that $0<a\le b < \infty$ are the constants from Definition~\ref{def:inputs}.

\begin{definition}[Palm criterion]
\label{def:criterion}
The layer marginals satisfy the Palm criterion if $an_\ell\le\lambda_\ell(v)\le bn_\ell$ and for almost every $v$, for every Borel set $S$ and every $t\in\mathbb R$,
\begin{align}
\E e^{t\cN_\ell(S)}
&\le
\exp\{\mu_\ell(S)(e^t-1)\},
\label{eq:palmcrit1}\\
\mu_\ell(S)-\rho^\ell
&\le
\mu_\ell^{!v}(S)
\le
\mu_\ell(S),
\label{eq:palmcrit2}\\
\E_\ell^{!v}e^{t|X\cap S|}
&\le
\exp\{\mu_\ell^{!v}(S)(e^t-1)\}.
\label{eq:palmcrit3}
\end{align}
The last two bounds hold for $\lambda_\ell(v)\,dv$-almost every $v$.
\end{definition}

The following proposition states that under Definition~\ref{def:criterion}, the hierarchy satisfies the properties from Definition~\ref{def:class} with high probability.

\begin{proposition}
\label{prop:coverage}
Suppose that Definition~\ref{def:criterion} holds and the parameters are chosen as in Appendix~\ref{app:dimension}. Then,
$$\PP(\cG_n^c)+\PP(\cH_n^c)=O(n^{-M}).$$
\end{proposition}

The Chernoff bounds in the next lemma are essential for the proof of Proposition~\ref{prop:coverage}.

\begin{lemma}
\label{lem:counts}
If a nonnegative integer-valued $Z$ with mean $\mu$ satisfies
$\E e^{tZ}\le e^{\mu(e^t-1)}$ for every real $t$, then
\begin{equation}
 \PP(Z=0)\le e^{-\mu},\quad
 \PP(Z<\mu/2)\le e^{-\mu/8},\quad
 \PP(Z\ge k)\le e^{-k\log2+\mu}.
 \label{eq:count-bounds}
\end{equation}
\end{lemma}
\begin{proof}
These standard Chernoff bounds follow from the assumed Poisson MGF domination, see e.g.~\cite{BoucheronLugosiMassart2013}, by taking $t\to-\infty$, $t=-\log2$ and $t=\log2$ respectively.
\end{proof}

To prove Proposition~\ref{prop:coverage}, we show that three counting events which imply the conditions in Definition~\ref{def:class} occur with high probability and conclude by using a union bound and~\eqref{eq:campbell}.

\begin{proof}[Proof of Proposition~\ref{prop:coverage}]
Fix a layer $\ell$ and a point $v$ for which the bounds in Definition~\ref{def:criterion} hold. Fix $\zeta = M+5$ and choose $$r_* = \left(\frac{C_{\rm pool}}{4bc_d}\right)^{1/d}.$$ Under $\PP_\ell^{!v}$, write
\[
\cN^{\rm out}:=|X\cap B(v,r_{\rm cov}h_\ell)|, \quad \cN^{\rm in}:= |X\cap B(v,r_*h_\ell)|, \quad \cN^{\rm sec}_j:=|X\cap B(v,r_*h_\ell)\cap\cC_j(v)|.
\]
If
\begin{align}\label{eq:innerouter ball condition}
   \cN^{\rm out} \ge m,\qquad \cN^{\rm in}\le m,\quad \text{ and } \quad \cN^{\rm sec}_j\ge 1\qquad\text{for every }j,
\end{align}
then the two conditions of Definition~\ref{def:class} hold: the first inequality implies (1), while the other two imply (2). Our goal now is to show that the inequalities in \eqref{eq:innerouter ball condition} hold with high probability. Since $h_\ell^d=\log n/n_\ell$ by equation~\eqref{eq:scale},
\begin{align}\label{eq:chernoff_prerequirement_1}
\mu = \E_\ell^{!v}\cN^{\rm out}\ge ac_dr_{\rm cov}^d\log n-\rho^\ell \ge \frac{ac_dr_{\rm cov}^d}{2}\log n \ge 2m
\end{align}
for all sufficiently large $n$. We apply Lemma~\ref{lem:counts}, via \eqref{eq:palmcrit3} and \eqref{eq:count-bounds}, and consider the parameter choices from Appendix~\ref{app:dimension} to obtain
\[
\PP_\ell^{!v}\bigl(\cN^{\rm out}<m\bigr)\le e^{-\mu/8} \le n^{-\zeta}.
\]
Similarly, we can show that
\begin{align}\label{eq:chernoff_prerequirement_2}
\E_\ell^{!v}|X\cap B(v,r_*h_\ell)|
\le \frac{C_{\rm pool}}4\log n,
\qquad
\E_\ell^{!v}\cN^{\rm sec}_j
\ge \frac{ap_jC_{\rm pool}}{4b}\log n-\rho^\ell.
\end{align}
Using Lemma~\ref{lem:counts} with $p_j\ge p_*$ and choosing $C_{\rm pool}$, we can show
\begin{align}\label{eq:estimate where Cpool is chosen}
    \PP_\ell^{!v}\bigl(\cN^{\rm in}>m\bigr)\le n^{-\zeta} \quad \text{ and } \quad \PP_\ell^{!v}\bigl(\cN^{\rm sec}_j=0\bigr)\le e^{\rho^\ell} n^{-\zeta}.
\end{align}
Define $$\cF_{\ell,v} := \{ \cN^{\rm out} <m \} \cup \{ \cN^{\rm in} >m \} \cup \bigcup\limits_{j=1}^K \{ \cN^{\rm sec}_j=0 \}$$ as the event that one of the inequalities in~\eqref{eq:innerouter ball condition} fails. By a union bound, we get
\begin{align*}
    \PP_\ell^{!v}(\cF_{\ell,v}) \le \PP_\ell^{!v}(\cN^{\rm out}<m) + \PP_\ell^{!v}(\cN^{\rm in}>m) + \sum\limits_{j=1}^K\PP_\ell^{!v}(\cN^{\rm sec}_j=0) \le (2+Ke^{\rho^\ell})n^{-\zeta}.
\end{align*}
Hence, using the reduced Campbell identity~\eqref{eq:campbell} yields
\begin{align*}
\PP(\cG_n^c)
&\le
\E\sum_{\ell=0}^L\sum_{v\in X_\ell}
\mathbf 1_{\cF_{\ell,v}} = \sum_{\ell=0}^L\int\limits_{\T^d} \lambda_\ell(v) \E_\ell^{!v}(\mathbf 1_{\cF_{\ell,v}}) dv \\
&\le \sum_{\ell=0}^L bn_\ell(2+Ke^{\rho^\ell})n^{-\zeta} \le\frac{b(2+Ke)}{1-\rho}\,n^{1-\zeta}=O(n^{-M-3}),
\end{align*}
where $K=n^{o(1)}$, since $d = o(\log n/\log\log n)$ and $\zeta = M+5$.
Finally, by using $aD \ge M$ and Lemma~\ref{lem:counts} we obtain
\[
\PP(X_L=\varnothing)\le e^{-an_L}\le n^{-aD}
\]
and since $\cH_n=\cG_n\cap\{X_L\ne\varnothing\}$, we get
\(
\PP(\cH_n^c)
\le
\PP(\cG_n^c)+\PP(X_L=\varnothing).
\)
The parameter choices therefore yield
\[
\PP(\cG_n^c)+\PP(\cH_n^c)=O(n^{-M}).\qedhere
\]
\end{proof}

\subsection{Verification for the three point-process laws}
Section~\ref{subsec:criterion} reduced the probabilistic analysis to the Palm criterion from Definition~\ref{def:criterion}. We now show that these estimates hold for the three input classes of Definition~\ref{def:inputs}. 

\begin{proposition}[Verification of the Palm criterion]
\label{prop:verification}
Poisson and determinantal input satisfy the Palm criterion from Definition~\ref{def:criterion}. Specifically, the difference between the original and the reduced Palm mean measures is zero for Poisson input and at most $\rho^\ell$ for determinantal input. For Cox input, the estimates in the Palm criterion hold conditionally on $F_n$, uniformly over all realizations with $a \le F_n \le b$.
\end{proposition}

Before we can prove Proposition~\ref{prop:verification}, we need the following lemma.

\begin{lemma}[Integrated determinantal Palm loss]
\label{lem:palm-loss}
The DPP kernel on layer $\ell$ is $\cK=\rho^\ell \cK_n$. Its reduced Palm kernel
at intensity-almost every $v$ is
\begin{equation}
 \cK^{!v}(x,y)=\cK(x,y)-\frac{\cK(x,v)\cK(v,y)}{\cK(v,v)}.
 \label{eq:palm-kernel}
\end{equation}
It is a positive contraction and for every Borel $S$,
\begin{equation}
 0\le\mu_\ell(S)-\mu_\ell^{!v}(S)
 =\frac{\int_S|\cK(x,v)|^2\,dx}{\cK(v,v)}\le\rho^\ell\le1.
 \label{eq:palm-loss}
\end{equation}
\end{lemma}
\begin{proof}[Proof of Lemma~\ref{lem:palm-loss}]
The thinning of the layers multiplies an $r$-point correlation density by $\rho^{\ell r}$,
which gives the kernel $\rho^\ell \cK_n$. The Schur complement of the
correlation determinant at $v$ gives \eqref{eq:palm-kernel}. By standard reduced Palm theory for Hermitian determinantal processes, see e.g. \cite{ShiraiTakahashi2003,MollerOReilly2021}, the reduced Palm process is determinantal with kernel $\cK^{!v}$ which satisfies $$0 \le \cK^{!v} \le \cK \le I.$$
Let \(\alpha_{n,i}\) and \(\phi_{n,i}\) denote the eigenvalues and orthonormal eigenfunctions of the integral operator with kernel \(\cK_n\). The spectral expansion is
\[ 
\cK(x,y) = \sum_i (\rho^\ell\alpha_{n,i}) \phi_{n,i}(x)\overline{\phi_{n,i}(y)}.
\] 
By orthonormality of the eigenfunctions, we get
\[ 
\int_{\T^d}|\cK(x,v)|^2\,dx = \sum_i (\rho^\ell\alpha_{n,i})^2 |\phi_{n,i}(v)|^2. 
\] 
Since \(0\le\alpha_{n,i}\le1\), we obtain 
\[ 
(\rho^\ell\alpha_{n,i})^2 \le \rho^\ell(\rho^\ell\alpha_{n,i}), 
\] 
and hence 
\[ 
\int_{\T^d}|\cK(x,v)|^2\,dx \le \rho^\ell \sum_i (\rho^\ell\alpha_{n,i}) |\phi_{n,i}(v)|^2 = \rho^\ell \cK(v,v), 
\] 
which implies \eqref{eq:palm-loss}.
\end{proof}

The three models require only minor variations on the same argument. For Poisson input, the required count and Palm estimates follow from the standard formulas for Poisson point processes. For determinantal input, local point counts can be represented as sums of independent Bernoulli variables, while Lemma~\ref{lem:palm-loss} shows that passing to the reduced Palm process changes the relevant means by at most $\rho^\ell$. For Cox input, we condition on the random environment, apply the corresponding Poisson estimates and then average over the environment. 

\begin{proof}[Proof of Proposition~\ref{prop:verification}]
Independent thinning preserves a Poisson law with intensity $n_\ell$.
The Poisson moment-generating function gives \eqref{eq:palmcrit1} with equality and the Mecke formula identifies its reduced Palm law with the unconditioned process~\cite{CoeurjollyMollerWaagepetersen2017}. Thus, \eqref{eq:palmcrit2} and \eqref{eq:palmcrit3} hold, since $\mu_\ell^{!v}(S)=\mu_\ell(S)$.

For determinantal input, the count \(\cN_\ell(S)\) is distributed as a sum of independent Bernoulli random variables whose parameters are the eigenvalues $(\theta_i)_i$ of 
\(\mathbf 1_S \cK \mathbf 1_S\), see e.g. \cite{Hough2006}. Hence, 
\[ 
\E e^{t\cN_\ell(S)} = \prod_i\bigl(1+\theta_i(e^t-1)\bigr) \le \exp\!\Big\{(e^t-1)\sum_i\theta_i\Big\} = \exp\bigl\{\mu_\ell(S)(e^t-1)\bigr\}, 
\] 
which shows \eqref{eq:palmcrit1}. By Lemma~\ref{lem:palm-loss} the reduced Palm process is again determinantal with a positive-contraction kernel whose mean measure differs by at most $\rho^\ell$ on every Borel set. Therefore, \eqref{eq:palmcrit2} holds and \eqref{eq:palmcrit3} follows with the same argument as \eqref{eq:palmcrit1}. 

For Cox input, condition on the environment \(F_n=f\). The Poisson argument applies with intensity \(n_\ell f\), uniformly over \(a\le f\le b\). Thus, the Palm criterion holds conditionally with $\mu_\ell^{!v}(S)=\mu_\ell(S)$, uniformly over admissible environments. Applying Proposition~\ref{prop:coverage} conditionally and then averaging establishes the required bounds for Cox input.
\end{proof}

\subsection{Proof of Theorem~\ref{thm:stochastic}}
To prove Theorem~\ref{thm:stochastic}, we combine the results of the previous subsections with Theorem~\ref{thm:deterministic}.

\begin{proof}[Proof of Theorem~\ref{thm:stochastic}] By Propositions~\ref{prop:verification} and~\ref{prop:coverage}, \eqref{eq:good-probability} holds; in particular $ \PP(\cH_n^c)=O(n^{-M}). $ On the event $\cH_n$, Theorem~\ref{thm:deterministic} gives 
\[ 
d(W_0,q)\le (1+\varepsilon)R_0(q) 
\] 
simultaneously for every query and every choice of entry point. This proves that~\eqref{eq:main-ann} holds with high probability for all vertices and layers. 
\end{proof}

\section{Proof of Theorem~\ref{thm:complexity}}
\label{sec:cost-proof}

Theorem~\ref{thm:deterministic} shows that the terminal vertex on each layer $\ell$ lies within distance $Bh_\ell$ of the query, as illustrated in Figure~\ref{fig:layer-descent}. As the greedy search only moves towards vertices that are closer to $q$, for $\ell \ge 1$ every vertex visited during the search at layer $\ell-1$ lies in the ball $B(q,Bh_\ell)$. Proving Theorem~\ref{thm:complexity} therefore reduces to a counting problem showing that the number of greedy hops is bounded by the number of points in a specific region. 

\subsection{Hops on non-top layers}

\begin{figure}[t]
\centering
\begin{tikzpicture}[
    scale=0.7,
    >=Stealth,
    every node/.style={font=\scriptsize}
]
\coordinate (q) at (0,0);
\draw[hsnwGray!55,fill=hsnwGray!7] (q) circle (2.6);
\draw[hsnwGray!65,fill=white] (q) circle (1.95);
\draw[hsnwGray!75,fill=white] (q) circle (1.3);
\draw[hsnwBlue,very thick,fill=white] (q) circle (0.65);
\fill[hsnwRed] (q) circle (2pt) node[below right] {$q$};

\node[hsnwGray] at (3.08,0.15) {$Bh_L$};
\node[hsnwGray] at (2.51,-0.45) {$Bh_{\ell+1}$};
\node[hsnwGray] at (1.60,0.45) {$Bh_\ell$};
\node[hsnwBlue] at (1.01,-0.1) {$Bh_0$};
\draw[hsnwBlue,very thick,->] (-2.3,1.55) -- (-1.7,1.15);
\draw[hsnwBlue,very thick,->] (-1.7,1.15) -- (-1.0,0.7);
\draw[hsnwBlue,very thick,->] (-1.0,0.7) -- (-0.45,0.25);
\node[align=center] at (-2.55,2.0) {descent through\\layers};

\end{tikzpicture}
\caption{The
final vertex on layer $\ell$ is at distance at most $Bh_\ell$ from $q$.  Repetition down to
layer $0$ yields that the final scale is $B(\log n/n)^{1/d}$.}
\label{fig:layer-descent}
\Description{Greedy descent through successive layers localizes the endpoint in smaller neighborhoods of the query.}
\end{figure}
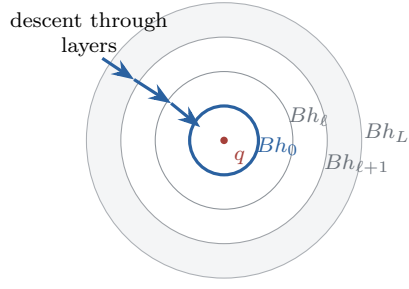

The next proposition shows that the number of hops on a single non-top layer has an exponentially decaying tail. 

\begin{proposition}
\label{prop:layer-cost}
Under the Palm criterion from Definition~\ref{def:criterion} and the parameter choices of Appendix~\ref{app:dimension}, there exists $\gamma > 0$ such that for each deterministic query, all $\ell<L$ and integers $k\ge1$, 
\[
 \PP(T_\ell\ge k,\cH_n)\le 2e^{-\gamma k}.
\]
In particular,
\begin{equation}
 \sup_{\ell<L}\E[T_\ell\1_{\cH_n}]
 \le\frac{2}{e^\gamma-1}=O(C_{\rm ann}^d).
 \label{eq:layer-mean}
\end{equation}
The same estimates hold for bounded Cox input.
\end{proposition}

\begin{lemma}
\label{lem:occupancy}
On $\cH_n$, for $\ell<L$, 
\begin{equation}
 T_\ell\le \cN_\ell(B(q,C_{\rm ann}R_{\ell+1}(q))),\qquad
 R_{\ell+1}(q)\le Bh_{\ell+1}
 ,\qquad C_{\rm ann}Bh_{\ell+1}<1/2.
 \label{eq:occupancy}
\end{equation}
\end{lemma}
\begin{proof}
Proposition~\ref{prop:geometry} shows that the layer's entry $W_{\ell+1}$ is
within $C_{\rm ann}R_{\ell+1}(q)$ of $q$. 
During the search on layer \(\ell\), every greedy hop strictly decreases the distance to \(q\). Therefore every vertex visited by a greedy hop lies in the ball \(B(q,C_{\rm ann}R_{\ell+1}(q))\). Since a greedy walk never revisits a vertex, the number of hops is at most the number of layer-\(\ell\) points contained in this ball. This proves the first inequality in~\eqref{eq:occupancy}. 

By definition of $R_\ell$, we also have that $R_{\ell+1}(q)\le d(W_{\ell+1},q).$ Combining this with Theorem~\ref{thm:deterministic} yields 
\[ 
R_{\ell+1}(q)\le d(W_{\ell+1},q)\le Bh_{\ell+1}. 
\] 
Finally, Appendix~\ref{app:dimension} shows that \( C_{\rm ann}Bh_{\ell+1}<\frac12.\)
\end{proof}

\begin{proof}[Proof of Proposition~\ref{prop:layer-cost}]
Set $\eta:=\frac{a\rho}{2bC_{\rm ann}^d}$ and $\xi:=\log2-\tfrac12$ and $\gamma:=\min\{\eta,\xi\}$.
For a given $k\ge1$, set
\[
 s_k=\Big(\frac{k}{2bc_dC_{\rm ann}^dn_\ell}\Big)^{1/d}.
\]
The radius \(s_k\) is chosen so that the expected number of points in a ball of radius \(C_{\mathrm{ann}}s_k\) is at most \(k/2\) in layer~\(\ell\).
We distinguish whether \(s_k\le Bh_{\ell+1}\). In the first case, both comparison balls are embedded and otherwise we use the deterministic radius \(Bh_{\ell+1}\) to bound the number of visited vertices.

If $s_k\le Bh_{\ell+1}$, then \eqref{eq:occupancy} gives 
\[
 \{T_\ell\ge k,\cH_n\}
 \subseteq\{R_{\ell+1}(q)>s_k\}
 \cup\{\cN_\ell(B(q,C_{\rm ann} s_k))\ge k\}.
\]
The first event is the event that $B(q,s_k)$ contains no point of $X_{\ell+1}$, so
its probability is at most
$$\exp(-an_{\ell+1}c_ds_k^d)=e^{-\eta k}.$$ The point count of the second event has mean at
most $$bn_\ell c_d(C_{\rm ann}s_k)^d=k/2,$$ so \eqref{eq:count-bounds} shows that the probability of the event is bounded by
$e^{-\xi k}$.

If $s_k>Bh_{\ell+1}$, the ball of radius $C_{\rm ann}s_k$ might wrap around the torus.  By Lemma~\ref{lem:occupancy}, every vertex visited during the search on layer \(\ell\) is in \(B(q,C_{\rm ann}Bh_{\ell+1})\). Since the number of hops is bounded by the number of layer-\(\ell\) points in this ball, \( T_\ell \le \cN_\ell(B(q,C_{\rm ann}Bh_{\ell+1})). \) Thus, 
\[ 
\{T_\ell\ge k,\cH_n\} \subseteq \{\cN_\ell(B(q,C_{\rm ann}Bh_{\ell+1}))\ge k\}. 
\] 
Because \(s_k>Bh_{\ell+1}\), 
\[ 
bn_\ell c_d(C_{\rm ann}Bh_{\ell +1})^d < bn_\ell c_d(C_{\rm ann}s_k)^d = k/2. 
\] 
Thus, the expected number of points in \(B(q,C_{\rm ann}Bh_{\ell+1})\) is strictly smaller than \(k/2\). Lemma~\ref{lem:counts} gives 
\[ 
\PP(T_\ell\ge k,\cH_n) \le \PP\!\left(\cN_\ell(B(q,C_{\rm ann}Bh_{\ell+1}))\ge k\right) \le e^{-\xi k}.
\] 

Summing the resulting tail over $k$ proves
\eqref{eq:layer-mean}. The Cox statement follows from uniform conditional
Poisson bounds.
\end{proof}

\subsection{Storage}
The hierarchy contains \(N_{\rm tot}\) layer incidences, each storing at most \(K\) outgoing arcs. Consequently, \(A_{\rm tot}\le KN_{\rm tot}\), so bounding \(N_{\rm tot}\) also controls graph storage. The next proposition shows that \(N_{\rm tot}\) is linear in \(n\) both in expectation and with exponentially high probability.
\begin{proposition}[Linear layer incidences]
\label{prop:storage}
For every admissible input there are fixed $C_0,c_0>0$ with
\[
 \PP(N_{\rm tot}>C_0n)\le e^{-c_0n},\qquad
 \E N_{\rm tot}\le\frac{bn}{1-\rho},\qquad
 \E N_{\rm tot}^2=O(n^2).
\]
\end{proposition}

\begin{proof}[Proof of Proposition~\ref{prop:storage}] Let \(Y_x=H_x+1\) and observe that $N_{\textrm{tot}}$ is bounded by \begin{equation}\label{eq:ntotbound}
N_{\rm tot}\le \widetilde N:=\sum_{x\in X_0}Y_x.   
\end{equation}
For \(0<t<\log(1/\rho)\), the moment generating function of \(Y_x\) is 
\[ 
g(t):=\E e^{tY_x} =\frac{(1-\rho)e^t}{1-\rho e^t}. 
\] 
Conditioning on \(X_0\) gives \( \E[e^{t\widetilde N}\mid X_0] = g(t)^{|X_0|}. \) Applying Equation~\eqref{eq:palmcrit1} to the whole torus gives 
\[ 
\E e^{t\widetilde N} = \E[g(t)^{|X_0|}] \le \exp\{bn(g(t)-1)\}. 
\] 
For Cox input, the same estimate holds conditionally on the environment, since the total directing mass is at most \(b\), and therefore also after averaging. Choosing \(C_0>b(g(t)-1)/t\) and applying Chernoff's inequality yields 
\[ 
\PP(\widetilde N>C_0n)\le e^{-c_0n} 
\] 
for some \(c_0>0\) and hence the same bound for \(N_{\rm tot}\). The first two moments of $Y_x$ are \(\E Y_x=(1-\rho)^{-1}\) and \(\E Y_x^2=(1+\rho)/(1-\rho)^2\). Using independence of the heights and~\eqref{eq:ntotbound}, we obtain
\[ 
\E \widetilde N = \E|X_0|\,\E Y_x \le \frac{bn}{1-\rho}, 
\] 
while 
\[ 
\E[\widetilde N^2\mid X_0] = |X_0|\E Y_x^2 + |X_0|(|X_0|-1)(\E Y_x)^2. 
\] 
The ordinary MGF bound also gives $$\mathbb E[|X_0|(|X_0|-1)]\le(\mathbb E|X_0|)^2\le b^2n^2,$$ by comparing the second-order expansions at zero. Taking expectations then yields $$\E\widetilde N^2=O(n^2). $$

For Cox input, the same calculation is performed conditionally on the environment and then averaged. Since \(N_{\rm tot}\le\widetilde N\), these estimates transfer directly to \(N_{\rm tot}\). Finally, each layer incidence contributes at most \(K\) outgoing arcs, so \( A_{\rm tot}\le K N_{\rm tot}, \) which gives the stated bounds. 
\end{proof}

\begin{proof}[Proof of Theorem~\ref{thm:complexity}]
Fix a deterministic query $q$. We decompose the hop count into three contributions: the top layer, the remaining layers, and the hop count on the  event $\cH_n^c$. As every point in each layer is visited at most once, $T_L\le|X_L|$ and $T_{\rm tot}\le N_{\rm tot}$. Hence,
\begin{align*}
 \E T_{\rm tot}
 &\le\E|X_L|+\sum_{\ell<L}\E[T_\ell\1_{\cH_n}]
       +\E[N_{\rm tot}\1_{\cH_n^c}]\\
 &\le(bD/\rho)\log n+LO(C_{\textrm{ann}}^d)
       +(\E N_{\rm tot}^2)^{1/2}\PP(\cH_n^c)^{1/2}
\end{align*}
Here, the first term is analyzed by~\eqref{eq:hierarchy}, the second term comes from Proposition~\ref{prop:layer-cost} and~\eqref{eq:layer-mean}, and the third term uses Cauchy-Schwarz.
Using $L=O(\log n)$ together with Propositions~\ref{prop:coverage} and~\ref{prop:storage} yields
\begin{align*}
 \E T_{\rm tot}
 &=O((D+C_{\textrm{ann}}^d)\log n)+O(n^{1-M/2})=O((D+C_{\textrm{ann}}^d)\log n)
\end{align*}
as $M>2$. Since \(C_{\textrm{ann}}\le1+\varepsilon\), we get $C_{\textrm{ann}}^d = O((1+\varepsilon)^d)$ and by Appendix~\ref{app:dimension}, we obtain $$D= O(\exp(\tfrac{1}{2}d\log(d)+O(d))),$$ so we can conclude with
\begin{align*}
 \E T_{\rm tot}
 &=O(\exp(\tfrac{1}{2}d\log(d)+O(d))\log n).
\end{align*}
Independent random queries are handled by conditioning on the query.
\end{proof}

\appendix
\section{Dimension-dependent parameter construction}
\label{app:dimension}
We verify the parameter choices required by Theorems~\ref{thm:stochastic} and~\ref{thm:complexity}  when $d=o(\log n/\log\log n).$ 

\begin{proof}[Verification of the parameter choices] 
For \(d=1\), the two directions \(\{\pm1\}\) define the required sector partition and all parameters are fixed constants. Assume \(d\ge2\). Fix \(0<\omega' < \omega_\varepsilon<\pi/3. \) Let \(\{u_1,\ldots,u_{K}\}\subset \mathbb S^{d-1}\) be a maximal \((\omega'/2)\)-separated set. The associated spherical Voronoi cells have maximum angular
diameter $\omega\le\omega'$, and each contains the spherical cap of angular radius \(\omega'/4\) around its center. Let \(v_d(r)\) denote the normalized surface measure of a spherical cap of radius \(r\). For fixed \(r\in(0,\pi/2)\), 
\[ 
v_d(r) \ge \frac{r}{2\pi}\sin^{d-2}(r/2) = \exp(-O_r(d)). 
\] 
Consequently, 
\[ 
K \le v_d(\omega'/4)^{-1} = \exp(O(d)), 
\] 
while every sector has measure at least 
\[ 
p_{*} \ge v_d(\omega'/4) = \exp(-O(d)). 
\] 
The proof of Proposition~\ref{prop:coverage} requires three local count events: the outer ball must contain at least \(m\) points, the inner ball must contain at most \(m\) points, and every inner sector must contain a point. We choose
\begin{align*}
    \zeta = M+5, \qquad C_{{\rm pool}} = \max\left\{1,\frac{\zeta}{\log 2 -1/4},\frac{4b\zeta}{ap_*}\right\}, \alpha = \max\{ 8C_{{\rm pool}},16\zeta \}
\end{align*}
and next
\[ 
r_{\rm cov}^d=\frac{\alpha}{ac_d}, \qquad B=\frac{r_{\rm cov}}{\cos\omega}, \qquad r_{*}^d=\frac{C_{{\rm pool}}}{4bc_d}.
\] 
These choices guarantee that the inequalities in equations~\eqref{eq:chernoff_prerequirement_1}, \eqref{eq:chernoff_prerequirement_2} and~\eqref{eq:estimate where Cpool is chosen} hold. 
The query-cost analysis uses comparison balls centered at the query. To apply the count estimates, these balls must remain embedded in the torus, so that their volume is exactly \(c_dr^d\). Choose a fixed constant 
\[ 
Q>\max\Big\{2,\frac{4}{\cos\omega'},\frac{2(1+\eps)}{\cos\omega'}\Big\}
\] 
and define 
\[ 
D = 1+a^{-1}(M+3)+2(Qr_{\rm cov})^d. 
\] 
As shown below, $D\log n = o(n)$. Hence, for sufficiently large $n$ the definition of $L$ in~\eqref{eq:hierarchy} gives $h_L \le D^{-1/d}$ and since $r_*<r_{\rm cov}$, the choice of $D$ ensures
\begin{align*}
    r_*h_L < r_{\rm cov}h_L < \frac{1}{2}, \qquad Bh_L<\frac{1}{4}, \qquad C_{\rm ann}Bh_{L}<\frac{1}{2}.
\end{align*}
To estimate the resulting dimension dependence, note that 
\[ 
c_d=\frac{\pi^{d/2}}{\Gamma(1+d/2)}. 
\] 
By Stirling's formula, 
\[ 
\log(c_d^{-1}) = \tfrac12 d\log d+O(d). 
\] 
Since 
\begin{align*}
r_{\rm cov}=(\alpha/a)^{1/d}c_d^{-1/d}\quad \text{ and } \quad \alpha=\exp(O(d)),
\end{align*}
we get
\[ 
r_{\rm cov}=O(\sqrt d), \qquad B=O(\sqrt d), \qquad r_{*}=O(\sqrt d). 
\] 
Furthermore, 
\[ 
D = \exp\Big( \tfrac12 d\log d+O(d) \Big). 
\] 
Finally, $ d\log d=o(\log n) $ implies that $K=n^{o(1)}$, $C_{{\rm pool}}=n^{o(1)}$, $D=n^{o(1)}$. In particular, we have $ D\log n=o(n)$, so for sufficiently large $n$, the definition of $L$ in~\eqref{eq:hierarchy} gives
\[
D\log n\le n_L<(D/\rho)\log n.
\]
Finally, Lemma~\ref{lem:counts} yields
\begin{align*}
\PP(X_L=\varnothing)\le e^{-an_L}\le n^{-aD}=O(n^{-M}).
\end{align*}
\end{proof}

\section*{Declaration of AI assistance}

LLMs were used to check the language and grammar, to check the proofs of the main results, and to generate both the illustrative visualizations and numerical experiments. The authors take full responsibility for the paper. In particular, the model formulation, results and proofs were all written by the authors before verification and cleanup by LLMs. 

\section*{Acknowledgements}

S. Bhamidi was partially supported by the National Science Foundation through grants DMS-2113662, DMS-2413928, and DMS-2434559, and by the NSF Research Training Group grant DMS-2134107.

\bibliographystyle{plain}
\bibliography{references}

\end{document}